\documentclass[11pt,english]{article}
\usepackage{lmodern}

\usepackage[T1]{fontenc}
\usepackage[latin9]{inputenc}
\usepackage{geometry}
\usepackage{dsfont}
\usepackage{amsmath}
\usepackage{amsthm}
\usepackage{amssymb}

\makeatletter
\numberwithin{figure}{section}
\numberwithin{equation}{section}
\theoremstyle{plain}
\newtheorem{thm}{\protect\theoremname}[section]
\theoremstyle{plain}
\newtheorem{cor}[thm]{\protect\corollaryname}
\theoremstyle{definition}
\newtheorem{defn}[thm]{\protect\definitionname}
\theoremstyle{plain}
\newtheorem{lem}[thm]{\protect\lemmaname}
\theoremstyle{plain}
\newtheorem{prop}[thm]{\protect\propositionname}
\theoremstyle{remark}
\newtheorem*{acknowledgement*}{\protect\acknowledgementname}

\usepackage{float,psfrag,epsfig,color,url,hyperref}
\usepackage{algorithm,algorithmic}
\usepackage{graphicx,relsize}
\usepackage{amssymb,amsfonts,amsmath,amsthm,amscd,dsfont,mathrsfs,mathtools,microtype,nicefrac,pifont}
\usepackage{upgreek}
\usepackage[dvipsnames]{xcolor}
\usepackage{epstopdf,bbm,enumitem}
\usepackage{dsfont,tikz}
\usepackage[mathscr]{euscript}
\usepackage[toc,page]{appendix}
\hypersetup{
  colorlinks,
  linkcolor={red!50!black},
  citecolor={blue!50!black},
  urlcolor={blue!80!black}
}
\usepackage{etoolbox}
\patchcmd{\thebibliography}
  {\settowidth}
  {\setlength{\itemsep}{0pt plus 0.1pt}\settowidth}
  {}{}
\apptocmd{\thebibliography}
  {\small}
  {}{}
\allowdisplaybreaks
\usepackage{custom}
\makeatletter
\renewcommand{\paragraph}{%
  \@startsection{paragraph}{4}%
  {\z@}{1.25ex \@plus 1ex \@minus .2ex}{-1em}%
  {\normalfont\normalsize\bfseries}%
}
\makeatother

\makeatother

\usepackage{babel}
\providecommand{\acknowledgementname}{Acknowledgement}
\providecommand{\corollaryname}{Corollary}
\providecommand{\definitionname}{Definition}
\providecommand{\lemmaname}{Lemma}
\providecommand{\propositionname}{Proposition}
\providecommand{\theoremname}{Theorem}

\begin{document}
\def\balign#1\ealign{\begin{align}#1\end{align}}
\def\baligns#1\ealigns{\begin{align*}#1\end{align*}}
\def\balignat#1\ealign{\begin{alignat}#1\end{alignat}}
\def\balignats#1\ealigns{\begin{alignat*}#1\end{alignat*}}
\def\bitemize#1\eitemize{\begin{itemize}#1\end{itemize}}
\def\benumerate#1\eenumerate{\begin{enumerate}#1\end{enumerate}}

\newenvironment{talign*}
 {\let\displaystyle\textstyle\csname align*\endcsname}
 {\endalign}
\newenvironment{talign}
 {\let\displaystyle\textstyle\csname align\endcsname}
 {\endalign}

\def\balignst#1\ealignst{\begin{talign*}#1\end{talign*}}
\def\balignt#1\ealignt{\begin{talign}#1\end{talign}}

\let\originalleft\left
\let\originalright\right
\renewcommand{\left}{\mathopen{}\mathclose\bgroup\originalleft}
\renewcommand{\right}{\aftergroup\egroup\originalright}

\def\Gronwall{Gr\"onwall\xspace}
\def\Holder{H\"older\xspace}
\def\Ito{It\^o\xspace}
\def\Nystrom{Nystr\"om\xspace}
\def\Schatten{Sch\"atten\xspace}
\def\Matern{Mat\'ern\xspace}

\def\tinycitep*#1{{\tiny\citep*{#1}}}
\def\tinycitealt*#1{{\tiny\citealt*{#1}}}
\def\tinycite*#1{{\tiny\cite*{#1}}}
\def\smallcitep*#1{{\scriptsize\citep*{#1}}}
\def\smallcitealt*#1{{\scriptsize\citealt*{#1}}}
\def\smallcite*#1{{\scriptsize\cite*{#1}}}

\def\blue#1{\textcolor{blue}{{#1}}}
\def\green#1{\textcolor{green}{{#1}}}
\def\orange#1{\textcolor{orange}{{#1}}}
\def\purple#1{\textcolor{purple}{{#1}}}
\def\red#1{\textcolor{red}{{#1}}}
\def\teal#1{\textcolor{teal}{{#1}}}

\def\mbi#1{\boldsymbol{#1}} 
\def\mbf#1{\mathbf{#1}}
\def\mrm#1{\mathrm{#1}}
\def\tbf#1{\textbf{#1}}
\def\tsc#1{\textsc{#1}}

\def\mbiA{\mbi{A}}
\def\mbiB{\mbi{B}}
\def\mbiC{\mbi{C}}
\def\mbiDelta{\mbi{\Delta}}
\def\mbif{\mbi{f}}
\def\mbiF{\mbi{F}}
\def\mbih{\mbi{g}}
\def\mbiG{\mbi{G}}
\def\mbih{\mbi{h}}
\def\mbiH{\mbi{H}}
\def\mbiI{\mbi{I}}
\def\mbim{\mbi{m}}
\def\mbiP{\mbi{P}}
\def\mbiQ{\mbi{Q}}
\def\mbiR{\mbi{R}}
\def\mbiv{\mbi{v}}
\def\mbiV{\mbi{V}}
\def\mbiW{\mbi{W}}
\def\mbiX{\mbi{X}}
\def\mbiY{\mbi{Y}}
\def\mbiZ{\mbi{Z}}

\def\textsum{{\textstyle\sum}} 
\def\textprod{{\textstyle\prod}} 
\def\textbigcap{{\textstyle\bigcap}} 
\def\textbigcup{{\textstyle\bigcup}} 

\def\reals{\mathbb{R}} 
\def\integers{\mathbb{Z}} 
\def\rationals{\mathbb{Q}} 
\def\naturals{\mathbb{N}} 
\def\complex{\mathbb{C}} 

\def\what#1{\widehat{#1}}

\def\twovec#1#2{\left[\begin{array}{c}{#1} \\ {#2}\end{array}\right]}
\def\threevec#1#2#3{\left[\begin{array}{c}{#1} \\ {#2} \\ {#3} \end{array}\right]}
\def\nvec#1#2#3{\left[\begin{array}{c}{#1} \\ {#2} \\ \vdots \\ {#3}\end{array}\right]} 

\def\maxeig#1{\lambda_{\mathrm{max}}\left({#1}\right)}
\def\mineig#1{\lambda_{\mathrm{min}}\left({#1}\right)}

\def\Re{\operatorname{Re}} 
\def\indic#1{\mbb{I}\left[{#1}\right]} 
\def\logarg#1{\log\left({#1}\right)} 
\def\polylog{\operatorname{polylog}}
\def\maxarg#1{\max\left({#1}\right)} 
\def\minarg#1{\min\left({#1}\right)} 
\def\Earg#1{\E\left[{#1}\right]}
\def\Esub#1{\E_{#1}}
\def\Esubarg#1#2{\E_{#1}\left[{#2}\right]}
\def\bigO#1{\mathcal{O}\left(#1\right)} 
\def\littleO#1{o(#1)} 
\def\P{\mbb{P}} 
\def\Parg#1{\P\left({#1}\right)}
\def\Psubarg#1#2{\P_{#1}\left[{#2}\right]}
\def\Trarg#1{\Tr\left[{#1}\right]} 
\def\trarg#1{\tr\left[{#1}\right]} 
\def\Var{\mrm{Var}} 
\def\Vararg#1{\Var\left[{#1}\right]}
\def\Varsubarg#1#2{\Var_{#1}\left[{#2}\right]}
\def\Cov{\mrm{Cov}} 
\def\Covarg#1{\Cov\left[{#1}\right]}
\def\Covsubarg#1#2{\Cov_{#1}\left[{#2}\right]}
\def\Corr{\mrm{Corr}} 
\def\Corrarg#1{\Corr\left[{#1}\right]}
\def\Corrsubarg#1#2{\Corr_{#1}\left[{#2}\right]}
\newcommand{\info}[3][{}]{\mathbb{I}_{#1}\left({#2};{#3}\right)} 
\newcommand{\staticexp}[1]{\operatorname{exp}(#1)} 
\newcommand{\loglihood}[0]{\mathcal{L}} 


\providecommand{\arccos}{\mathop\mathrm{arccos}}
\providecommand{\dom}{\mathop\mathrm{dom}}
\providecommand{\diag}{\mathop\mathrm{diag}}
\providecommand{\tr}{\mathop\mathrm{tr}}
\providecommand{\card}{\mathop\mathrm{card}}
\providecommand{\sign}{\mathop\mathrm{sign}}
\providecommand{\conv}{\mathop\mathrm{conv}} 
\def\rank#1{\mathrm{rank}({#1})}
\def\supp#1{\mathrm{supp}({#1})}

\providecommand{\minimize}{\mathop\mathrm{minimize}}
\providecommand{\maximize}{\mathop\mathrm{maximize}}
\providecommand{\subjectto}{\mathop\mathrm{subject\;to}}

\def\openright#1#2{\left[{#1}, {#2}\right)}

\ifdefined\nonewproofenvironments\else
\ifdefined\ispres\else
 
\fi
\fi
\makeatletter
\@addtoreset{equation}{section}
\makeatother
\def\theequation{\thesection.\arabic{equation}}

\newcommand{\cmark}{\ding{51}}

\newcommand{\xmark}{\ding{55}}

\newcommand{\eq}[1]{\begin{align}#1\end{align}}
\newcommand{\eqn}[1]{\begin{align*}#1\end{align*}}
\renewcommand{\Pr}[1]{\mathbb{P}\left( #1 \right)}
\newcommand{\Ex}[1]{\mathbb{E}\left[#1\right]}

\newcommand{\matt}[1]{{\textcolor{Maroon}{[Matt: #1]}}}
\newcommand{\kook}[1]{{\textcolor{blue}{[Kook: #1]}}}
\definecolor{OliveGreen}{rgb}{0,0.6,0}
\newcommand{\sv}[1]{{\textcolor{OliveGreen}{[Santosh: #1]}}}

\global\long\def\on#1{\operatorname{#1}}%

\global\long\def\bw{\mathsf{Ball\ walk}}%
\global\long\def\sw{\mathsf{Speedy\ walk}}%
\global\long\def\gw{\mathsf{Gaussian\ walk}}%
\global\long\def\ps{\mathsf{Proximal\ sampler}}%
\global\long\def\dw{\mathsf{Dikin\ walk}}%

\global\long\def\chr{\mathsf{Coordinate\ Hit\text{-}and\text{-}Run}}%
\global\long\def\har{\mathsf{Hit\text{-}and\text{-}Run}}%
\global\long\def\gc{\mathsf{Gaussian\ cooling}}%
\global\long\def\ino{\mathsf{\mathsf{In\text{-}and\text{-}Out}}}%
\global\long\def\tgc{\mathsf{Tilted\ Gaussian\ cooling}}%
\global\long\def\PS{\mathsf{PS}}%
\global\long\def\psunif{\mathsf{PS}_{\textup{unif}}}%
\global\long\def\psexp{\mathsf{PS}_{\textup{exp}}}%
\global\long\def\psann{\mathsf{PS}_{\textup{ann}}}%
\global\long\def\psgauss{\mathsf{PS}_{\textup{Gauss}}}%
\global\long\def\eval{\mathsf{Eval}}%
\global\long\def\mem{\mathsf{Mem}}%

\global\long\def\O{O}%
\global\long\def\Otilde{\widetilde{O}}%
\global\long\def\Omtilde{\widetilde{\Omega}}%

\global\long\def\E{\mathbb{E}}%
\global\long\def\Z{\mathbb{Z}}%
\global\long\def\P{\mathbb{P}}%
\global\long\def\N{\mathbb{N}}%

\global\long\def\R{\mathbb{R}}%
\global\long\def\Rd{\mathbb{R}^{d}}%
\global\long\def\Rdd{\mathbb{R}^{d\times d}}%
\global\long\def\Rn{\mathbb{R}^{n}}%
\global\long\def\Rnn{\mathbb{R}^{n\times n}}%

\global\long\def\psd{\mathbb{S}_{+}^{d}}%
\global\long\def\pd{\mathbb{S}_{++}^{d}}%

\global\long\def\defeq{\stackrel{\mathrm{{\scriptscriptstyle def}}}{=}}%

\global\long\def\veps{\varepsilon}%
\global\long\def\lda{\lambda}%
\global\long\def\vphi{\varphi}%
\global\long\def\K{\mathcal{K}}%

\global\long\def\half{\frac{1}{2}}%
\global\long\def\nhalf{\nicefrac{1}{2}}%
\global\long\def\texthalf{{\textstyle \frac{1}{2}}}%
\global\long\def\ltwo{L^{2}}%

\global\long\def\ind{\mathds{1}}%
\global\long\def\op{\mathsf{op}}%
\global\long\def\ch{\mathsf{Ch}}%
\global\long\def\kls{\mathsf{KLS}}%
\global\long\def\ts{\mathsf{Ts}}%
\global\long\def\hs{\textup{HS}}%
\global\long\def\ls{\textup{LS}}%

\global\long\def\cpi{C_{\mathsf{PI}}}%
\global\long\def\clsi{C_{\mathsf{LSI}}}%
\global\long\def\cch{C_{\mathsf{Ch}}}%
\global\long\def\clch{C_{\mathsf{logCh}}}%
\global\long\def\cexp{C_{\mathsf{exp}}}%
\global\long\def\cgauss{C_{\mathsf{Gauss}}}%

\global\long\def\chooses#1#2{_{#1}C_{#2}}%

\global\long\def\frob{\on F}%

\global\long\def\vol{\on{vol}}%

\global\long\def\sym{\on{sym}}%

\global\long\def\law{\on{law}}%

\global\long\def\tr{\on{tr}}%

\global\long\def\diag{\on{diag}}%

\global\long\def\diam{\on{diam}}%

\global\long\def\poly{\on{poly}}%

\global\long\def\polylog{\on{polylog}}%

\global\long\def\Diag{\on{Diag}}%

\global\long\def\inter{\on{int}}%

\global\long\def\esssup{\on{ess\,sup}}%

\global\long\def\proj{\on{Proj}}%

\global\long\def\e{\mathrm{e}}%

\global\long\def\id{\mathrm{id}}%

\global\long\def\supp{\on{supp}}%

\global\long\def\spanning{\on{span}}%

\global\long\def\rows{\on{row}}%

\global\long\def\cols{\on{col}}%

\global\long\def\rank{\on{rank}}%

\global\long\def\T{\mathsf{T}}%

\global\long\def\bs#1{\boldsymbol{#1}}%

\global\long\def\eu#1{\EuScript{#1}}%

\global\long\def\mb#1{\mathbf{#1}}%

\global\long\def\mbb#1{\mathbb{#1}}%

\global\long\def\mc#1{\mathcal{#1}}%

\global\long\def\mf#1{\mathfrak{#1}}%

\global\long\def\ms#1{\mathscr{#1}}%

\global\long\def\mss#1{\mathsf{#1}}%

\global\long\def\msf#1{\mathsf{#1}}%

\global\long\def\textint{{\textstyle \int}}%
\global\long\def\Dd{\mathrm{D}}%
\global\long\def\D{\mathrm{d}}%
\global\long\def\grad{\nabla}%
 
\global\long\def\hess{\nabla^{2}}%
 
\global\long\def\lapl{\triangle}%
 
\global\long\def\deriv#1#2{\frac{\D#1}{\D#2}}%
 
\global\long\def\pderiv#1#2{\frac{\partial#1}{\partial#2}}%
 
\global\long\def\de{\partial}%
\global\long\def\lagrange{\mathcal{L}}%
\global\long\def\Div{\on{div}}%

\global\long\def\Gsn{\mathcal{N}}%
 
\global\long\def\BeP{\textnormal{BeP}}%
 
\global\long\def\Ber{\textnormal{Ber}}%
 
\global\long\def\Bern{\textnormal{Bern}}%
 
\global\long\def\Bet{\textnormal{Beta}}%
 
\global\long\def\Beta{\textnormal{Beta}}%
 
\global\long\def\Bin{\textnormal{Bin}}%
 
\global\long\def\BP{\textnormal{BP}}%
 
\global\long\def\Dir{\textnormal{Dir}}%
 
\global\long\def\DP{\textnormal{DP}}%
 
\global\long\def\Exp{\textnormal{Exp}}%
 
\global\long\def\Gam{\textnormal{Gamma}}%
 
\global\long\def\GEM{\textnormal{GEM}}%
 
\global\long\def\HypGeo{\textnormal{HypGeo}}%
 
\global\long\def\Mult{\textnormal{Mult}}%
 
\global\long\def\NegMult{\textnormal{NegMult}}%
 
\global\long\def\Poi{\textnormal{Poi}}%
 
\global\long\def\Pois{\textnormal{Pois}}%
 
\global\long\def\Unif{\textnormal{Unif}}%

\global\long\def\bpar#1{\bigl(#1\bigr)}%
\global\long\def\Bpar#1{\Bigl(#1\Bigr)}%

\global\long\def\abs#1{|#1|}%
\global\long\def\babs#1{\bigl|#1\bigr|}%
\global\long\def\Babs#1{\Bigl|#1\Bigr|}%

\global\long\def\snorm#1{\|#1\|}%
\global\long\def\bnorm#1{\bigl\Vert#1\bigr\Vert}%
\global\long\def\Bnorm#1{\Bigl\Vert#1\Bigr\Vert}%

\global\long\def\sbrack#1{[#1]}%
\global\long\def\bbrack#1{\bigl[#1\bigr]}%
\global\long\def\Bbrack#1{\Bigl[#1\Bigr]}%

\global\long\def\sbrace#1{\{#1\}}%
\global\long\def\bbrace#1{\bigl\{#1\bigr\}}%
\global\long\def\Bbrace#1{\Bigl\{#1\Bigr\}}%

\global\long\def\Abs#1{\left\lvert #1\right\rvert }%
\global\long\def\Par#1{\left(#1\right)}%
\global\long\def\Brack#1{\left[#1\right]}%
\global\long\def\Brace#1{\left\{  #1\right\}  }%

\global\long\def\inner#1{\langle#1\rangle}%
 
\global\long\def\binner#1#2{\left\langle {#1},{#2}\right\rangle }%

\global\long\def\norm#1{\lVert#1\rVert}%
\global\long\def\onenorm#1{\norm{#1}_{1}}%
\global\long\def\twonorm#1{\norm{#1}_{2}}%
\global\long\def\infnorm#1{\norm{#1}_{\infty}}%
\global\long\def\fronorm#1{\norm{#1}_{\text{F}}}%
\global\long\def\nucnorm#1{\norm{#1}_{*}}%
\global\long\def\staticnorm#1{\|#1\|}%
\global\long\def\statictwonorm#1{\staticnorm{#1}_{2}}%

\global\long\def\mmid{\mathbin{\|}}%

\global\long\def\otilde#1{\widetilde{O}(#1)}%
\global\long\def\wtilde{\widetilde{W}}%
\global\long\def\wt#1{\widetilde{#1}}%

\global\long\def\KL{\msf{KL}}%
\global\long\def\dtv{d_{\textrm{\textup{TV}}}}%
\global\long\def\FI{\msf{FI}}%
\global\long\def\tv{\msf{TV}}%
\global\long\def\TV{\msf{TV}}%

\global\long\def\cov{\on{cov}}%
\global\long\def\var{\on{Var}}%
\global\long\def\ent{\on{Ent}}%

\global\long\def\cred#1{\textcolor{red}{#1}}%
\global\long\def\cblue#1{\textcolor{blue}{#1}}%
\global\long\def\cgreen#1{\textcolor{green}{#1}}%
\global\long\def\ccyan#1{\textcolor{cyan}{#1}}%
\global\long\def\yk#1{\textcolor{red}{\textsf{[YK: #1]}}}%
\global\long\def\yb#1{\textcolor{blue}{\textsf{[yb: #1]}}}%

\global\long\def\iff{\Leftrightarrow}%
 
\global\long\def\textfrac#1#2{{\textstyle \frac{#1}{#2}}}%

\global\long\def\Expo{\textnormal{Expo}}%
\global\long\def\Tr{\on{Tr}}%
\global\long\def\onu{\bar{\nu}}%
\global\long\def\intk{\inter\K}%
\global\long\def\ncal{\mathcal{N}}%
\global\long\def\svec{\operatorname{svec}}%
\global\long\def\tvec{\operatorname{vec}}%
\global\long\def\del{\partial}%
\global\long\def\ovec{\operatorname{ovec}}%

\title{Beyond the $d^{2.5}$-mixing bound for Dikin walks on polytopes\\
\date{}\author{Yunbum Kook\\ Georgia Tech\\  \texttt{yb.kook@gatech.edu}}}
\maketitle
\begin{abstract}
Inspired by interior-point methods (IPM) for structured convex optimization,
Kannan and Narayanan introduced the Dikin walk for sampling uniformly
from polytopes in 2009. As in IPMs, the Dikin walk is affine-invariant,
and its convergence is governed by the barrier geometry used to define
its local proposal. They showed that the Dikin walk with the logarithmic
barrier for a polytope in $\mathbb{R}^{d}$ with $m$ linear inequalities
mixes in $md$ iterations. In 2017, Chen, Dwivedi, Wainwright, and
Yu improved this to $d^{2.5}$ using a Lewis-weight barrier, and conjectured
that the correct mixing time should be $d^{2}$. 

We make progress toward this conjecture by improving the previous
$d^{2.5}$-mixing bound. For exponential sampling over a polytope,
we prove that the Dikin walk with a scaled Lee--Sidford metric mixes
from a warm start in $d^{2.25}$ iterations. This also yields an improved
cold-start complexity via a known annealing framework. The main technical
ingredient is improved average self-concordance of the Lee--Sidford
metric, which gives high acceptance probability for the Metropolis
filter along a random Dikin proposal. While previous analyses were
effectively limited to second-order control due to technical difficulties,
we develop a principled higher-order analysis. The proof combines
a selective higher-order expansion of recursive bottleneck terms,
a moving orthonormal-frame calculus for higher derivatives of the
Lewis weights, and Wiener-chaos decompositions via multiple stochastic
integrals to control the resulting Gaussian polynomials.
\end{abstract}
\thispagestyle{empty}\tableofcontents{}

\section{Introduction}

\setcounter{page}{1} We study the complexity of sampling from a high-dimensional
polytope. Throughout the paper, for $A\in\R^{m\times d}$ and $b\in\R^{m}$,
let $\K:=\{x\in\Rd:Ax\geq b\}$ be a bounded full-dimensional polytope
with $m$ linear inequalities. The goal of this problem is to output
a sample whose law is close to the uniform distribution (more generally,
an exponential distribution) over $\K$. In theoretical terms, this
problem has rich connections to convex optimization and convex-body
sampling \cite{DFK91random,LS93random,KLS97random,LV06hit,CV18Gaussian,JLLV26reducing,KVZ26INO,KZ25Renyi},
and on the applied side, it has found important applications in studying
human metabolic networks in systems biology \cite{TSF13community,HCTFV17chrr,KLSV22sampling,KLSV23condition}.

This is an important and well-studied special case of the more general
problem of \emph{log-concave sampling} with only zeroth-order oracle
access to the distribution \cite{LV07geometry,KV25sampling,KV25faster,KV26zeroLC}.
The best-known approaches currently have complexity that grows roughly
quadratically in the dimension \cite{KLS97random,KVZ24INO,KVZ26INO}
from a warm start (i.e., when already close to a target log-concave
distribution) and sub-cubic \cite{KV25faster,KV26zeroLC} from a cold
start, after an appropriate affine transformation called (near-)isotropic
rounding \cite{JLLV26reducing}. Here, each step involves a zeroth-order
query (i.e., checking whether a point $x$ is in $\K$), which for
a polytope has arithmetic complexity $O(md)$ (for checking all the
inequalities)\footnote{In the case of a polytope, by amortizing membership queries and using
fast matrix multiplication, this can be reduced to $md^{c_{1}+0.5}$
for rounding and $md^{c_{2}}$ for each sample thereafter, where $c_{1}<3.2$
and $c_{2}<2.26$ with current matrix multiplication bounds \cite{JLLV26reducing}.}.

A natural question is whether the \emph{structure} of a polytope,
and in particular its explicit inequality description, can be exploited
in the design of the sampler, rather than only through oracle queries.
General log-concave samplers treat the inequality description only
as a black-box oracle. Hence, a more ``structural'' approach would
leverage the additional information to accelerate mixing of polytope
samplers. To motivate principled approaches to this question and highlight
its deep connection to optimization theory, we first turn to the corresponding
optimization problem.

\paragraph{Sampling and optimization.}

Sampling and optimization can be related through the correspondence
between (1) sampling from a Gibbs distribution $\D\pi(x)\propto e^{-V(x)}\,\D x$
and (2) minimizing the potential $V:\Rd\to\R\cup\{\infty\}$. Indeed,
one can use sampling to solve the optimization problem via simulated
annealing \cite{KGV83simulated,HJJ03simulated,KV06simulated}---sampling
from the annealed distributions $\D\pi_{t}\propto e^{-V/t}\,\D x$
while gradually decreasing the temperature $t>0$: as $t$ becomes
small, the mass of $\pi_{t}$ concentrates near the minimizers of
$V$. Conversely, ideas from optimization have inspired the design
and analysis of sampling algorithms. A prominent example is the connection
between gradient flow and Langevin dynamics: the latter can be interpreted
as a Wasserstein gradient flow over the space of probability measures
\cite{JKO98variational}.

Under the above correspondence, the optimization counterpart of exponential
sampling over a polytope $\K$ is \emph{linear programming} (LP):
minimizing a linear function over a set of linear inequalities. Since
LP is a special case of convex optimization, one might first consider
applying a general-purpose optimization method, such as gradient descent.
However, such methods typically depend on the conditioning of the
problem instance and may struggle on ill-conditioned problems. This
brings us back to the structural question above: when solving LP,
how can one exploit the explicit inequality description of the polytope,
rather than treating it merely as a black-box convex body?

\paragraph{Interior-point methods (IPM).}

This question is precisely what led to the development of interior-point
methods. In 1967, Dikin introduced an interior-point method based
on the Dikin ellipsoid \cite{Dikin1967iterative}, a local ellipsoidal
approximation of the feasible region induced by the logarithmic barrier.
Following the subsequent breakthroughs of Karmarkar \cite{Karmarkar84polytime},
Renegar \cite{Renegar88polytime}, and Nesterov and Nemirovski \cite{NN94ipm},
the IPM literature has evolved into a general theory based on \emph{self-concordant}
barriers. In this theory, the iteration complexity is controlled by
the \emph{barrier parameter}. Roughly speaking, self-concordance imposes
a regularity condition on the third-order derivative of the barrier,
ensuring that its local Hessian geometry does not change too rapidly,
while the barrier parameter quantifies the global quality of the barrier.
For example, for a polytope $\K$, the logarithmic barrier $\phi(x)=-\sum_{i=1}^{m}\log(a_{i}^{\T}x-b_{i})$
has barrier parameter $\nu=O(m)$.

A crucial feature of IPM is \emph{affine-invariance}. Unlike general-purpose
optimizers, their performance is not affected by the conditioning
of the problem instance, but by the geometry induced by the chosen
barrier. Nesterov and Nemirovski showed that IPM based on a self-concordant
barrier with parameter $\nu$ requires $O(\nu^{1/2})$ iterations.
Thus, using the logarithmic barrier gives an $O(m^{1/2})$-iteration
method for LP. A major direction in IPM theory has been to reduce
this dependence on the number $m$ of constraints, since $m$ can
be inflated by nearly redundant inequalities without changing the
intrinsic geometry of the feasible region\footnote{In principle, this $m$-dependence can be avoided: the universal barrier
\cite{NN89selfconcordant} and the entropic barrier \cite{BE15entropic}
are $d$-self-concordant \cite{LY21universal,Chewi23entropic}, and
hence yield $O(d^{1/2})$-iteration IPM. However, these barriers are
generally difficult to compute or implement efficiently. }.

This motivated the search for efficiently computable barriers with
smaller parameters. In 1989, Vaidya's breakthrough \cite{Vaidya96convex}
constructed such a barrier with parameter $O(\sqrt{md})$, based on
leverage scores of the constraint matrix. Finally in 2014, Lee and
Sidford obtained an efficiently computable barrier with parameter
$\Otilde(d)$ based on \emph{Lewis weights} of the constraint matrix
\cite{LS14pathfinding,LS19solving}, which leads to $\Otilde(d^{1/2})$-iteration
IPM. Thus, in the optimization theory of LP, the right choice of a
barrier removes the linear dependence on the number of inequalities.

\paragraph{IPM for sampling: the Dikin walk.}

The connection above suggests that the IPM perspective can also be
used to address the structural question for sampling. In the setting
of uniform sampling from a polytope, which may be viewed as the ``random''
feasibility for linear inequalities, Kannan and Narayanan initiated
an interior-point theory for sampling by proposing the $\dw$ \cite{KN09Dikin,KN12random}
in 2009. This walk is based on the classical \emph{Dikin ellipsoid}
from convex optimization: given a local metric $g:\inter\K\to\pd$
and a radius $r>0$, let $\msf N_{g}^{r}(x):=\msf N(x,\frac{r^{2}}{d}\,g(x)^{-1})$
be the Gaussian proposal with mean $x$ and covariance matrix $\frac{r^{2}}{d}\,g(x)^{-1}$.
The Dikin ellipsoid of radius $r$ at $x$ is defined as
\[
\msf D_{g}(x,r):=\{y\in\Rd:\norm{y-x}_{g(x)}\leq r\}\,,\qquad\norm v_{g(x)}^{2}=\norm v_{x}^{2}:=v^{\T}g(x)v\quad\text{for }v\in\Rd\,.
\]

Starting from $x_{0}\sim\pi_{0}$, the $\dw$ with target distribution
$\pi$ and local metric $g$ repeats the following steps: for $i=0,1,\dots$,
\begin{enumerate}
\item Sample $y\sim\msf N_{g}^{r}(x_{i})$.
\item If $y\notin\K$, then $x_{i+1}\gets x_{i}$. If $y\in\K$, set $x_{i+1}\gets y$
with probability 
\begin{equation}
\min\Bbrace{1,\frac{\pi(y)\sqrt{\det g(y)}\exp(-\frac{d}{2r^{2}}\,\norm{y-x_{i}}_{g(y)}^{2})}{\pi(x_{i})\sqrt{\det g(x_{i})}\exp(-\frac{d}{2r^{2}}\,\norm{y-x_{i}}_{g(x_{i})}^{2})}}\,,\label{eq:acceptance-probability}
\end{equation}
and $x_{i+1}\gets x_{i}$ otherwise.
\end{enumerate}
Kannan and Narayanan used the Hessian of the logarithmic barrier for
the local metric $g$ for uniform sampling $\pi\propto\ind_{\K}$.
This makes the walk affine-invariant, and therefore avoids the usual
preprocessing step of putting the body into near-isotropic position.

They showed that, from an $M$-warm start\footnote{An initial distribution $\mu$ is said to be $M$-warm with respect
to $\pi$ if $\mu(S)\le M\pi(S)$ for every measurable set $S$.}, the $\dw$ with the log-barrier mixes in $\Otilde(md\log\frac{M}{\veps})$
iterations to total variation (TV) distance $\varepsilon$ from the
uniform distribution; see also \cite{SV16mixing} for a simplified
proof. Moreover, each step of the log-barrier $\dw$ can be implemented
in amortized $\Otilde(md)$ time~\cite{LLV20strong}. Subsequently,
Narayanan extended this viewpoint to convex bodies equipped with self-concordant
barriers, further developing the interpretation of the $\dw$ as a
randomized IPM for sampling and optimization \cite{Narayanan16randomized}.

More recently, Kook and Vempala improved the iteration complexity
to $\Otilde(md\log\frac{1}{\veps})$ from a cold start (i.e., no warm
start) \cite{KV24ipm}. A related recent direction is due to Jiang
and Chen~\cite{JC24regularized}, who studied \emph{regularized}
$\dw$ for sampling log-concave distributions truncated to polytopes.
Their results include a soft-threshold $\dw$ with warm-start mixing
$\Otilde(md+\kappa d)$ for log-concave and log-smooth targets with
condition number $\kappa$.

\paragraph{Toward dimension-square mixing.}

For polytopes with many constraints (i.e., $m\gg d$), the $\dw$
may suffer from slow mixing. This mirrors the central question from
the optimization IPM: can the linear dependence on $m$ be improved
by choosing a better barrier?

There have been a series of works toward this goal, closely paralleling
the development of improved barriers in optimization. Gustafson and
Narayanan showed that by using a geometry related to the classical
John ellipsoid, one can obtain an iteration complexity polynomial
only in the dimension, though with a large exponent, namely $d^{27}$
\cite{GN23john}. Subsequently, Chen, Dwivedi, Wainwright, and Yu
showed that the $\dw$ with a variant of the Lewis-weights metric
mixes in $\Otilde(d^{2.5}\log\frac{M}{\veps})$ steps from an $M$-warm
start~\cite{CDW18MCMC}, where the Lewis-weights metric can be viewed
as an efficiently implementable proxy for the John ellipsoid geometry.
In particular, they left the conjecture that the mixing time could
be $\Otilde(d^{2}\polylog\frac{M}{\veps})$. This is the main motivation
of the present paper: \emph{can the iteration complexity of the $\dw$
be reduced from the log-barrier bound $\Otilde(md)$ to $\Otilde(d^{2}\polylog m)$?}

There are several reasons why a dimension-square mixing bound is the
natural target. First, in the optimization IPM theory, the barrier
parameter can be reduced from $m$ to $\Otilde(d)$ by replacing the
log-barrier with the Lee--Sidford barrier. Second, for general-purpose
samplers such as the $\bw$ \cite{KLS97random} and $\ino$ \cite{KVZ26INO,KV25sampling},
a dimension-square rate arises naturally once the target distribution
has been put in isotropic position. For example, the mixing bounds
for the aforementioned samplers scale as $d^{2}\,\norm{\cov\pi}$,
where the largest eigenvalue $\norm{\cov\pi}$ of the covariance matrix
measures the geometric skewness of the target distribution. Hence,
those samplers require (near-)isotropic rounding (i.e., $\norm{\cov\pi}=O(1)$),
after which they mix in $d^{2}$ steps from an $O(1)$-warm start.
Since the $\dw$ is affine-invariant, it does not necessitate this
global rounding step. Intuitively, the local metric used by the $\dw$
plays the role of \emph{implicitly} rounding the body around the current
point. From this perspective, the $\dw$ with a properly chosen metric
should behave as though the target were locally isotropic, making
$d^{2}$-mixing the fundamental target. 

\paragraph{Earlier attempts.}

Since then, there have been several attempts toward the conjectured
$d^{2}$-mixing bound. The challenge is that we need delicate analysis
of the change of local metric: in bounding the acceptance probability
of the proposal $y$ from the current point $x$, one must not only
show that a proposal is well-behaved in the forward ellipsoid centered
at the current point (i.e., $\msf N(x,\frac{r^{2}}{d}\,g(x)^{-1})$),
but also control the corresponding reverse proposal density at the
proposal (i.e., $\msf N(y,\frac{r^{2}}{d}\,g(y)^{-1})$). This is
referred to as average self-concordance (ASC) in the terminology of
\cite{KV24ipm}. This is where some previous claims of dimension-square
mixing have gaps that cannot be readily fixed.

In 2020, Laddha, Lee, and Vempala~\cite{LLV20strong} introduced
extensions of the classical theory of self-concordance, including
strong self-concordance and symmetry, as part of a general framework
for analyzing the $\dw$. Their work originally claimed a $d^{2}$-iteration
bound, but this claim leaves open an important Metropolis-filter issue.
Specifically, after drawing a proposal $y$ from the forward Dikin
ellipsoid centered at the current point $x$, the analysis also needs
to verify that $x$ lies in the Dikin ellipsoid centered at $y$ with
high probability. This step was initially missed, and a direct way
of fixing this issue yields a $d^{3}$-bound \cite{V26}.

In 2024, \cite{GKM+24Dikin} claimed $d^{2}$-mixing of the $\dw$
with Lewis weights. In the proof sketch of Lemma D.5, they assert
that \cite[Proposition 7]{SV16mixing} can be used similarly for the
Lewis weights, but without the complete justification\footnote{\cite{GKM+24Dikin} uses closeness of Lewis weights at the base point
$x$ and a random proposal $z$, obtaining $\abs{\norm{z-x}_{z}^{2}-\norm{z-x}_{x}^{2}}=O(1)\,\norm{z-x}_{x}^{2}\lesssim r^{2}\norm h^{2}/d\approx r^{2}$.
Since ASC needs the $O(r^{2}/d)$-scale, this closeness argument is
not sufficient for their claim. The log-barrier argument of \cite{SV16mixing}
achieves this smaller scale through the Gaussian-polynomial concentration.
For Lewis weights, however, the corresponding expansion contains $w_{z}$
and Taylor remainders along the random segment $[x,x+\eta h]$, so
the similar polynomial argument does not apply since \textbf{the estimated
point $w_{z}$ has dependence on $h$}. In fact, this is exactly the
main challenge in showing ASC.}. This is exactly the ASC challenge underlying the Metropolis filter
that prevents prior works from improving $d^{2.5}$-mixing of \cite{CDW18MCMC}.
Despite several plausible approaches, a complete proof of $d^{2}$-mixing
for a Lewis-weight-based $\dw$ has remained elusive.

\subsection{Results\label{subsec:results}}

We make the first progress in nearly a decade toward the conjectured
$d^{2}$-mixing bound for the weighted $\dw$ of Chen, Dwivedi, Wainwright,
and Yu. We show that the $\dw$ with a smaller $d^{1/4}$ scaling
of the Lee--Sidford metric has the warm-start iteration complexity
$\Otilde(d^{2.25})$.
\begin{thm}
[Sampling from warm start] \label{thm:target-mixing} Let $\K=\{x\in\Rd:Ax\geq b\}$
be a full-dimensional bounded polytope, where $A\in\R^{m\times d}$
has full column rank and no all-zero rows, and $b\in\R^{m}$. Let
$g_{0}$ be the unscaled Lee--Sidford metric \eqref{eq:LS-metric},
$L=Cd^{1/4}\polylog(md)$ for a universal constant $C>0$, and $g:=Lg_{0}$.
Consider the $\dw$ with local metric $g$, a sufficiently small $\Theta(1)$
radius $r$, an initial distribution $\pi_{0}$, and a target exponential
distribution $\pi$ over $\K$. Let $\pi_{n}$ be the law of the $n$-th
iterate of the $\half$-lazy\footnote{It stays at the current point (i.e., $x_{i+1}\gets x_{i}$) with probability
$1/2$.} $\dw$ for $n\in\mbb N$. Then, for any $\veps>0$, we obtain $\chi^{2}(\pi_{n}\mmid\pi)\leq\varepsilon$
after
\[
N=\widetilde{O}\bpar{d^{2.25}\polylog m\log\frac{\chi^{2}(\pi_{0}\mmid\pi)}{\veps}}\quad\text{iterations}.
\]
\end{thm}

This result assumes a warm start, and a natural next question is how
efficiently such a warm start can be generated. One simple approach
is to initialize from an easy distribution, such as the uniform distribution
over a small ball contained in $\K$. However, this typically introduces
an extra factor of $d$ in the iteration complexity. Thus, obtaining
a sharp cold-start guarantee requires a more sophisticated approach. 

Such a framework was developed in \cite{KV24ipm}. Their approach
constructs a sequence of annealing distributions and repeatedly uses
a Dikin-walk sampler to move from one distribution to the next. In
the polytope setting, the annealing path has the form proportional
to $e^{-\phi/t}\cdot\ind_{\K}$, where $\phi$ is an appropriate self-concordant
barrier and the temperature parameter $t$ is gradually increased.
The schedule is chosen so that consecutive distributions are sufficiently
close, allowing the previous distribution to serve as an $O(1)$-warm
start for the next phase. Using this framework with the Lewis-weights
metric, they recovered the known $\Otilde(d^{2.5})$ warm-start bound
and obtained an $\Otilde(d^{23/8})$ cold-start iteration bound, giving
the first sub-cubic iteration complexity for this problem.

By plugging our improved warm-start guarantee into the annealing framework
of~\cite{KV24ipm}, we obtain an improved cold-start bound of $d^{41/16}\approx d^{2.56}$.
\begin{cor}
[Sampling from cold start] \label{cor:ws-gen}Given $\varepsilon>0$,
there exists an algorithm that generates a sample $X$ using $\Otilde(d^{41/16}\polylog m\log\frac{1}{\veps})$
iterations of the $\dw$ such that $\dtv(\law X,\pi)\leq\veps$.
\end{cor}

Finally, we note that the above bounds are stated in terms of Dikin-walk
iterations. The corresponding arithmetic complexities can be obtained
by multiplying the iteration bounds by the per-step cost of computing
Lewis weights, which contributes a factor of approximately $md^{\omega-1}$
per iteration, up to polylogarithmic factors, where $\omega<2.37$
is the matrix multiplication exponent.  

\subsection{Technical overview\label{subsec:technical}}

We now summarize the technical ingredients of the main results. In
this overview, we ignore any polylogarithmic factors in $m$. In the
optimization IPM, the barrier parameter $\nu$ of a self-concordant
barrier is sufficient to bound the iteration complexity of the IPM
as $O(\nu^{1/2})$. However, in the sampling counterpart, self-concordance
and barrier parameters alone are not enough to provide a tight mixing
guarantee of the $\dw$.

\paragraph{Mixing analysis through self-concordance.}

The framework of \cite[Theorem 1]{KV24ipm} shows that if the local
metric $g:\inter\K\to\pd$ is strong self-concordant (SSC), lower
trace self-concordant (LTSC), average self-concordant (ASC), and $\bar{\nu}$-symmetric
(see \S\ref{subsec:Dikin-walk-prelim}), then the $\dw$ with $\Theta(1)$-radius,
an initial distribution $\pi_{0}$, and a target exponential distribution
$\pi$ mixes in $O(d\bar{\nu}\log\frac{\chi^{2}(\pi_{0}\mmid\pi)}{\veps})$
iterations to be $\veps$-close to $\pi$ in $\chi^{2}$ (Lemma~\ref{prop:gcdw-framework}).

Recall the Lee--Sidford (LS) metric. For $x\in\inter\K$, let $S_{x}:=\Diag(Ax-b)\in\R^{m\times m}$
and $A_{x}:=S_{x}^{-1}A$. For $p\asymp\polylog m$, let $W_{x}=\Diag w_{x}$
be the $\ell_{p}$-Lewis weights of $A_{x}$, defined as $w_{x}=\sigma(W_{x}^{1/2-1/p}A_{x})$,
where $\sigma(B)$ denotes the leverage scores of $B$. The unscaled
LS metric is defined as $g_{0}(x):=A_{x}^{\T}W_{x}^{1-2/p}A_{x}$.
It was already shown in \cite{LLV20strong} that $g_{0}$ satisfies
$\onu=\Otilde(d)$, SSC, and LTSC. Hence, the only missing condition
for applying the mixing framework is ASC.

\paragraph{ASC of the Lee--Sidford metric.}

Recall that given a local metric $g$, the radius-$r$ Dikin proposal
at $x$ is $z\sim\msf N_{g}^{r}(x):=\msf N(x,\frac{r^{2}}{d}\,g(x)^{-1})$.
The Metropolis filter then compares this forward proposal density
with the reverse proposal density from $z$ back to $x$, and a crucial
part is to bound the acceptance probability \eqref{eq:acceptance-probability}
from below, by showing the closeness of the Gaussian proposal at $x$
and $z$. This desired property is referred to as \emph{average self-concordance}
in \cite{KV24ipm}. Informally, ASC says that a random Dikin proposal
does not significantly change its squared local length when the metric
is evaluated at the current point $x$ and the proposal $z$. Precisely,
a metric $g$ is ASC if for every $\veps>0$ there is a dimension-free
radius $r_{\veps}>0$ such that for every $x\in\inter\K$ and every
$r\le r_{\veps}$, 
\[
\P_{z\sim\msf N_{g}^{r}(x)}\bpar{\abs{\norm{z-x}_{g(z)}^{2}-\norm{z-x}_{g(x)}^{2}}\le\frac{2\veps r^{2}}{d}}\ge1-\veps\,.
\]

If we can show ASC of the LS metric, then we would immediately obtain
the $d^{2}$-mixing as desired. However, existing analyses \cite{CDW18MCMC,KV24ipm,JC24regularized}
prove ASC of the LS metric \emph{only after scaling it by roughly
$d^{1/2}$}. Since scaling the metric by a factor $L$ also scales
the symmetry parameter by $L$, this changes $\bar{\nu}=\Otilde(d)$
into $\bar{\nu}=\Otilde(d^{3/2})$, and then the general framework
gives $d\onu=d^{2.5}$ mixing bound. Namely, the extra $d^{1/2}$
in the previous mixing bound comes from the amount of scaling needed
to ensure constant-radius ASC.

In this paper, we improve this scaling from $d^{1/2}$ to $d^{1/4}$,
showing that the unscaled LS metric $g_{0}$ satisfies the ASC condition
at radius $d^{-1/8}$. If $g=Lg_{0}$, then a radius-$r$ proposal
for $g$ is the same (in law) as a radius-$r/L^{1/2}$ proposal for
$g_{0}$. Hence, taking $L=\Otilde(d^{1/4})$ converts the radius-$d^{-1/8}$
ASC for $g_{0}$ into constant-radius ASC for $g$. This scaling increases
the symmetry parameter only to $\Otilde(d^{5/4})$, and the mixing
framework yields $\Otilde(d\cdot d^{5/4})=\Otilde(d^{9/4})$ iterations.
Thus, the proof of Theorem~\ref{thm:target-mixing} is reduced to
proving this sharper ASC estimate for the unscaled LS metric.

\paragraph{Prior approaches to ASC.}

We now describe the ASC proof. By affine invariance, fix a base point
$x=0$ and normalize the unscaled LS metric so that $g_{0}(x)=I_{d}$.
A proposal from the unscaled metric can be written as $z=x+\eta h$,
where $h\sim\msf N(0,I_{d})$ and $\eta=r/\sqrt{d}$. Define the path
\[
F(t):=h^{\T}g_{0}(x+th)h\,,\qquad0\le t\le\eta\,.
\]
Since $F(0)=\norm{z-x}_{g_{0}(x)}^{2}/\eta^{2}$ and $F(\eta)=\norm{z-x}_{g_{0}(z)}^{2}/\eta^{2}$,
ASC follows from a high-probability bound on $|F(\eta)-F(0)|\leq2\veps$.

A straightforward approach would be a direct Taylor expansion of $F(\eta)$
around $\eta=0$. For the $\dw$ with Lewis weights, one may use the
first-order expansion, $F(\eta)-F(0)=\eta F'(t_{*})$ for some $t_{*}\in(0,\eta)$.
However, technical computations lead to $\sup_{t\leq\eta}|F'(t)|\lesssim d$,
and this requires a $d$-scaling of the metric (which yields $\Otilde(d^{3})$-mixing
of the $\dw$).

Chen et al.\ \cite{CDW18MCMC} expanded $F$ up to the second-order
terms, $F(\eta)-F(0)=\eta F'(0)+\frac{\eta^{2}}{2}\,F''(t_{*})$.
We note that at the base point $x=0$, the first-order term $F'(0)$
is now a \emph{polynomial in Gaussian variables} $h$. As already
used in \cite{SV16mixing,CDW18MCMC}, the Gaussian polynomial $P(h)$
concentrates around its $L^{2}$-norm $(\E_{\gamma}[P(h)^{2}])^{1/2}$
(see Lemma~\ref{lem:conc-gaussian-poly}), and this allows one to
save an extra $d^{1/2}$-factor in $F'(0)$. Indeed, after substantial
technical computations, \cite{CDW18MCMC} managed to show that with
high probability, $|F'(0)|\lesssim d^{1/2}$ and $\sup_{t\le\eta}|F''(t)|\lesssim d^{3/2}$.
In this case, the second-order contribution is $\eta^{2}d^{3/2}=r^{2}d^{1/2}$.
To make this $O(\veps)$, the radius $r$ should be reduced to $d^{-1/4}$,
which introduces the $d^{1/2}$-scaling of their metric for a constant-radius
ASC.

One could guess from this pattern that $F^{(k)}(0)\lesssim d^{k/2}$
and $\sup_{t\le\eta}\abs{F^{(k)}(t)}\lesssim d^{(k+1)/2}$ with high
probability. To implement this idea, one could consider expanding
further and attempt to obtain better bounds. However, a brute-force
higher-order Taylor expansion quickly becomes unmanageable: $F''(0)$
already contains several complicated Gaussian polynomials involving
derivatives of Lewis weights, while higher derivatives introduce many
additional terms and require calculus beyond the known first-order
Lewis-weight calculus in \cite{LS19solving}. This suggests that we
need a more organized approach and principled tools to handle this
technical problem.

\paragraph{Technical contributions in the improved ASC proof.}

The improvement is not obtained by simply differentiating the LS metric
a few more times. The main difficulty is to make the higher-order
expansion give the small number of bottleneck terms, while keeping
both the pathwise Lewis-weight calculus and the base-point Gaussian-polynomial
estimates tractable. Our proof introduces three organizing ideas.
First, instead of Taylor expanding all derivatives of $F$, we isolate
a recursive bottleneck chain $H_{k}(t)$ and expand only this chain;
all other differentiated terms are treated as controlled pathwise
terms. Second, we formulate the higher-order LS calculus through a
moving orthonormal frame for the column space of a half of the LS
metric. This removes irrelevant rotations of the column space and
converts derivatives of the Lewis-weight matrix into controlled frame
and row-map estimates. Third, for the base-point terms, we replace
direct high-moment Gaussian calculations by a Hermite-polynomial decomposition,
implemented through multiple stochastic integrals. This gives a canonical
orthogonal decomposition of the high-order Gaussian polynomials and
reduces their $L^{2}$ estimates to structured tensor-norm bounds.

\subparagraph{(1) Isolation of bottleneck terms.}

The first idea is a selective expansion that avoids the combinatorial
explosion of higher-order Taylor terms of $F$. Along the path $x_{t}=x+th$,
write $A_{t}:=A_{x_{t}}$, $W_{t}:=W_{x_{t}}$, and let $N_{t}$ denote
the Lewis-weight derivative matrix from Lemma~\ref{lem:DWh}. Set
$\alpha:=1-\frac{2}{p}$ and $\beta:=\alpha/2$, and define $(v^{\circ2})_{i}:=v_{i}^{2}$
for a vector $v$ (and similarly for a matrix) and
\[
s_{t}:=A_{t}h\,,\qquad u_{t}:=W_{t}^{\beta}s_{t}\,,\qquad v_{t}:=W_{t}^{1/2}s_{t}\,,\qquad q_{t}:=W_{t}^{-1/2}u_{t}^{\circ2}\,.
\]
With this notation, $F(t)=\norm{u_{t}}^{2}$. The Lewis-weight derivative
formula gives a decomposition of the first derivative of the form
$F'(t)=-2\beta H_{1}(t)-2E_{1}(t)$, where $H_{1}(t):=q_{t}^{\T}N_{t}v_{t}$
and $E_{1}(t):=\inner{s_{t},W_{t}^{\alpha}s_{t}^{\circ2}}$.

A simple but crucial observation is that not every term produced by
differentiating $F'$ is a bottleneck. To isolate the bottleneck terms,
define $H_{k}(t):=q_{t}^{\T}N_{t}^{(k-1)}v_{t}$ for $k\ge1$, where
$N_{t}^{(j)}$ denotes the $j$-th derivative of $N_{t}$ along the
path. Differentiating $H_{k}$ gives $H_{k}'(t)=G_{k+1}(t)+H_{k+1}(t)$,
where $G_{k+1}(t)$ collects the terms in which the derivative lands
on $q_{t}$ or $v_{t}$, while $H_{k+1}(t)$ is the single term in
which the derivative lands on $N_{t}^{(k-1)}$. The point of this
decomposition is that the $G$-terms are controllable from the same
pathwise estimates used at the previous level, whereas the new $H$-term
contains the next derivative of $N_{t}$ and is the only bottleneck
term. Thus, we only expand $H_{1},H_{2},H_{3},H_{4}$, rather than
all terms appearing in the derivatives of $F$. The resulting bound
is as follows:
\begin{equation}
\begin{aligned}|F(\eta)-F(0)| & \lesssim\eta\,|E_{1}(0)+\beta H_{1}(0)|+\eta^{2}\,\bpar{\sup_{t\le\eta}|E_{1}'(t)|+\sup_{t\le\eta}|G_{2}(t)|+|H_{2}(0)|}\\
 & \quad+\eta^{3}\,\bpar{\sup_{t\le\eta}|G_{3}(t)|+|H_{3}(0)|}+\eta^{4}\sup_{t\le\eta}\bpar{|G_{4}(t)|+|H_{4}(t)|}\,.
\end{aligned}
\label{eq:taylor-intro}
\end{equation}
The proof then splits into two tasks: (1) pathwise bounds for the
good terms $E_{1}'$, $G_{2}$, $G_{3}$, $G_{4}$, and $H_{4}$,
and (2) Gaussian-polynomial estimates for the base-point terms $E_{1}(0)$
and $H_{i}(0)$ for $i\in[3]$.

\subparagraph{(2) Pathwise part.}

Recall preliminary calculus for Lewis weights in \cite{LS19solving}:
for $c_{p}:=1-2/p$, $\Lambda_{x}:=W_{x}-P_{x}^{\circ2}$, $\bar{\Lambda}_{x}:=W_{x}^{-1/2}\Lambda_{x}W_{x}^{-1/2}$,
they obtained $W_{x,h}':=\Dd W_{x}[h]=-\Diag(W_{x}^{1/2}N_{x}W_{x}^{1/2}s_{x,h})$
(see Lemma~\ref{lem:DWh} for details). For the pathwise part, we
extend this computation to develop a higher-order calculus for Lewis
weights in Lemma~\ref{lem:good-event}. The main difficulty in computing
higher-order derivatives of $W_{x}$ arises from $N_{x}$. Indeed,
its derivative was computed together with several additional matrices
in \cite[Lemma 36 and 37]{LS19solving}, so a direct computation quickly
becomes unwieldy.

The second idea is to carry out this higher-order Lewis-weight calculus
in a moving orthonormal frame. We work with $\widehat{B}_{t}:=W_{t}^{\beta}A_{t}$
and choose a smooth orthonormal frame $U_{t}$ for $\cols\widehat{B}_{t}$
satisfying the condition $U_{t}^{\T}U_{t}'=0$ (see \S\ref{app:moving-frame-calculus}).
This frame removes the arbitrary rotation of the basis inside the
column space, streamlining computations for downstream calculus.

Letting $u_{t,i}$ and $\phi_{t,i}:=(u_{t,i}\otimes u_{t,i})/w_{t,i}^{1/2}$
be the $i$-th rows of $U_{t}$ and $\Phi_{t}$, respectively, we
rewrite $\bar{\Lambda}_{t}=I-\Phi_{t}\Phi_{t}^{\T}$. Then, bounds
on $N_{t}'$, $N_{t}''$, and $N_{t}'''$ reduce to bounds on $\Phi_{t}'$,
$\Phi_{t}''$, and $\Phi_{t}'''$, which in turn reduce to bound on
$U_{t}',U_{t}'',U_{t}'''$. Using the identities for $U_{t}',U_{t}'',U_{t}'''$
(Lemma~\ref{lem:path-normal-frame}) together with row-wise estimates
for the map $u\mapsto u\otimes u/\norm u$ in \S\ref{app:good-event-proof},
we prove the good-event bounds in Lemma~\ref{lem:good-event}. This
is the step that turns the higher-order Lewis-weight calculus into
a systematic set of matrix estimates. On this event, uniformly for
all $t\in[0,\eta]$, we control the relevant slack directions, Lewis-weight
scores, and derivatives of $N_{t}$.

Consequently, Lemma~\ref{lem:good-term-bounds} gives $\abs{E_{1}'(t)},\abs{G_{2}(t)}\lesssim d$,
$\abs{G_{3}(t)}\lesssim d^{3/2}$, $\abs{G_{4}(t)}\lesssim d^{2}$,
and $\abs{H_{4}(t)}\lesssim d^{5/2}$. The recursive decomposition
is important here: once $G_{k+1}$ is separated from $H_{k+1}$ in
$H_{k}'$, bounds on $G_{k+1}$ terms immediately follow from the
same pathwise bounds on $q_{t}$, $q_{t}'$, $v_{t}$, $v_{t}'$,
and $N_{t}^{(k-1)}$. Thus, these terms are controlled inductively
rather than by expanding all derivatives of $F$ directly. The remaining
term $H_{4}$ is bounded by Cauchy--Schwarz as $\abs{H_{4}(t)}=\abs{q_{t}^{\T}N_{t}'''v_{t}}\le\norm{q_{t}}\,\norm{N_{t}'''}\,\norm{v_{t}}\lesssim d^{5/2}$.
This is the residual bottleneck of our argument. In \eqref{eq:taylor-intro},
it contributes $\eta^{4}d^{5/2}=r^{4}d^{1/2}$, which is exactly what
limits our ASC radius to $r=\Otilde(d^{-1/8})$.

\subparagraph{(3) Base-point part.}

For the base-point terms, we exploit the fact that they are Gaussian
polynomials in direction $h$, in order to save an extra $d^{1/2}$
in their high-probability bound. Once their $L^{2}$-norms are bounded,
the standard concentration inequality for Gaussian polynomials (Lemma~\ref{lem:conc-gaussian-poly})
converts those $L^{2}$-estimates into high-probability bounds. 

The first base-point term is the cubic polynomial $E_{1}(0)+\beta H_{1}(0)$.
In previous analyses \cite{CDW18MCMC,KV24ipm}, controlling this cubic
term was one of the most involved calculations. We streamline this
computation by using the Gaussian Poincar\'e inequality, recovering
the previous bound of $\norm{E_{1}(0)+\beta H_{1}(0)}_{L^{2}}\lesssim d^{1/2}$.
Hence, the first term in \eqref{eq:taylor-intro} contributes $\eta d^{1/2}=O(r)$.

The most technical part of this paper is to estimate the $L^{2}$-norms
of the quartic and quintic terms $H_{2}(0)$ and $H_{3}(0)$. These
are Gaussian polynomials of degrees four and five, so a direct $L^{2}$-norm
computation would require expanding eighth and tenth Gaussian moments.
Such an expansion produces too many tensor terms. It is not clear
how to organize unwieldy computations by hand, and this is where we
need a more principled tool to keep these computations under control.

Our approach is to decompose them into orthogonal Hermite polynomials
(equivalently, \emph{Wiener chaoses} \cite{Nualart06Malliavin,Nualart19malliavin})
and then use orthogonality in estimating the $L^{2}$-norm. This replaces
a non-canonical high-moment expansion by an orthogonal decomposition,
so that downstream tasks are substantially streamlined. Then, the
MSI (\emph{multiple stochastic integral}) formalism provides the algebra
needed for this decomposition.

For a symmetric $L^{2}$-function $f:\R_{+}^{n}\to\R$ and Brownian
motion $(B_{t})_{t\geq0}$, the MSI of $f$ (see Definition~\ref{def:MSI})
is defined as
\[
I_{n}(f):=n!\,\int_{0}^{\infty}\int_{0}^{t_{n}}\cdots\int_{0}^{t_{2}}f(t_{1},\ldots,t_{n})\,\D B_{t_{1}}\cdots\D B_{t_{n}}\,.
\]
Note that this is the usual It\^o integral when $n=1$. One can represent
base Gaussian polynomials as MSIs, and then leverage known results
on MSI (e.g., Theorem~\ref{thm:wiener-decom}, Lemma~\ref{lem:msi-product},
and \eqref{eq:orthogonality}) to estimate the $L^{2}$-norm in a
streamlined and organized way. Also, a known isometry between symmetric
tensors and functions naturally induces the MSI of symmetric tensors
(see \eqref{eq:isometry}). Through the MSI isometry (Lemma~\ref{lem:tensor-isometry}),
$L^{2}$-norms of those MSIs of tensors can be related to the Frobenius
norm of those tensors. In summary, \emph{the $L^{2}$ estimates of
Gaussian polynomials are reduced to bounding the norm of induced tensors}.

We illustrate how this toolkit is used in bounding $\norm{H_{2}(0)}_{2}$.
By simple algebra, one can rewrite it as $H_{2}(0)=\sum_{i}w_{i}^{1/2}(h^{\T}K_{i}h)(h^{\T}S_{i}h)$
for some matrices $K_{i}$ and $S_{i}$. Using the identities $h^{\T}K_{i}h=1+I_{2}(K_{i})$
and $h^{\T}S_{i}h=\tr S_{i}+I_{2}(S_{i})$, we can represent $H_{2}(0)$
in terms of tensor MSIs as $H_{2}(0)=\sum_{i}w_{i}^{1/2}\{1+I_{2}(K_{i})\}\{\tr S_{i}+I_{2}(S_{i})\}$.
With the product formula for MSIs (Lemma~\ref{lem:msi-product})
and orthogonality \eqref{eq:orthogonality}, this decomposes into
orthogonal terms as $H_{2}(0)=I_{4}(T_{4})+I_{2}(T_{2})+T_{0}$, where
the $T_{i}$ are induced tensors of degree $i$. By the MSI isometry
(Lemma~\ref{lem:tensor-isometry}), the $L^{2}$-norm is reduced
to bounding $\norm{T_{4}}_{\frob}$, $\norm{T_{2}}_{\frob}$, and
$\abs{T_{0}}$. The remaining estimates exploit the fact $P^{\circ2}\preceq W$
and Lewis-weight derivative bounds, yielding $\|H_{2}(0)\|_{L^{2}}\lesssim d$.
The quintic term $H_{3}(0)$ is handled by the same strategy, with
second-derivative estimates for $N_{t}$ replacing the first-derivative
estimates, and gives $\|H_{3}(0)\|_{L^{2}}\lesssim d^{3/2}$. These
two estimates contribute $\eta^{2}d=O(r^{2})$ and $\eta^{3}d^{3/2}=O(r^{3})$
in \eqref{eq:taylor-intro}, respectively.

Combining the pathwise estimates with the base-point Gaussian-polynomial
bounds gives, with high probability, 
\[
\abs{F(\eta)-F(0)}=\Otilde(r+r^{2}+r^{3}+r^{4}d^{1/2}).
\]
Choosing $r\asymp d^{-1/8}$ makes the RHS sufficiently small, proving
ASC for the unscaled LS metric at radius $d^{-1/8}$. Scaling the
metric by $d^{1/4}$ then gives constant-radius ASC with symmetry
parameter $d^{5/4}$, and applying the mixing theorem of~\cite{KV24ipm}
proves Theorem~\ref{thm:target-mixing}.

\paragraph{Warm-start generation.}

We now sketch the proof of the cold-start guarantee in Corollary~\ref{cor:ws-gen}.
No new mixing argument is needed for this step, and we simply use
the better scaling of the LS metric and the annealing framework of
\cite{KV24ipm} as a black box.

Theorem 2 of \cite{KV24ipm} states that if the local metric $g$
is SSC, LTSC, ASC, and $\onu$-symmetric with a barrier parameter
$\nu$, then given $\veps>0$ and an exponential distribution $\pi$,
their annealing scheme can generate a sample $X$ such that $\dtv(\law X,\pi)\leq\veps$
using $\Otilde(d\max(d,\onu,\nu^{1/2}\bar{\nu}^{3/4})\log\frac{1}{\veps})$
iterations of the $\dw$. Since we now know that the scaled LS metric
$d^{1/4}g_{0}$ (previously $d^{1/2}g_{0}$) satisfies all these properties
with $\nu,\onu=\Otilde(d^{5/4})$, plugging this into the same annealing
framework gives $d\,(d^{5/4})^{5/4}=d^{41/16}$.

\subsection{Discussion}

We close this section by suggesting a potential route toward the conjectured
$d^{2}$-mixing. Recall that the bottleneck term in the proof is
$H_{j+1}(t)=q_{t}^{\T}N_{t}^{(j)}v_{t}$. Our argument establishes
pathwise control of this chain up to $j=3$, giving the terminal bound
$\abs{H_{4}(t)}\lesssim d^{5/2}$. Its contribution is $\eta^{4}d^{5/2}=r^{4}d^{1/2}$,
which forces the ASC radius $r=\Otilde(d^{-1/8})$ and hence the $d^{1/4}$
scaling of the metric.

If one could extend the pathwise estimates to (1) $\abs{q_{t}^{\T}N_{t}^{(j)}v_{t}}\lesssim d^{1+j/2}$
for $0\le j\le k$ (namely, $\norm{N_{t}^{(j)}}\lesssim d^{j/2}$
since $\norm{q_{t}},\norm{v_{t}}\lesssim d^{1/2}$ in Lemma~\ref{lem:good-event}),
and (2) the base-point Gaussian estimates to $\norm{q_{x}^{\T}N_{x}^{(j)}v_{x}}_{L^{2}}\lesssim d^{(j+1)/2}$
for $0\le j\le k-1$, then the mixing framework would yield $\Otilde(d^{2+1/(k+1)})$
warm-start mixing. The present paper carries out this program for
$k=3$, giving $d^{2+1/4}=d^{9/4}$. Therefore, further progress toward
the $d^{2}$-conjecture amounts to developing higher-order versions
of the two estimates proved here. We view the present paper as isolating
these two problems as the main remaining obstacles to the conjectured
$d^{2}$-mixing for polytope sampling within our approach. 

\section{Preliminaries\label{sec:sampling-framework}}

We assume throughout that $A\in\R^{m\times d}$ has full column rank
and no all-zero rows. For $x\in\inter\K$, let $a_{i}$ be the $i$-th
row of $A$ and write $S_{x}:=\Diag(a_{i}^{\T}x-b_{i})\in\R^{m\times m}$
and $A_{x}:=S_{x}^{-1}A\in\R^{m\times d}$.

\subsection{A mixing framework\label{subsec:Dikin-walk-prelim}}

\paragraph{Self-concordance for sampling.}

Recall the self-concordance terminology of \cite[Definition~1.1]{KV24ipm}.
The basic regularity condition is self-concordance (SC): for a local
metric $g$, this means 
\[
-2\norm h_{g(x)}\,g(x)\preceq\Dd g(x)[h]\preceq2\norm h_{g(x)}\,g(x)\qquad\text{for }x\in\intk\text{ and }h\in\Rd\,.
\]
The conditions SSC and LTSC are the stronger Frobenius and trace versions
used in the mixing framework; for completeness their precise forms
are recalled in \S\ref{app:self-concordance-definitions}. The Lee--Sidford
metric is SSC and LTSC. The condition proved in this paper is ASC,
so we spell it out explicitly.
\begin{defn}[{Average self-concordance, {\cite[Definition~1.1]{KV24ipm}}}]
\label{def:asc} A matrix function $g:\inter\K\to\pd$ is called
average self-concordant (ASC) if for any $\veps>0$ there exists dimension-free
$r_{\veps}>0$ such that, for every $x\in\inter\K$, if $r\leq r_{\veps}$,
then 
\[
\P_{z\sim\msf N_{g}^{r}(x)}\bpar{\abs{\norm{z-x}_{g(z)}^{2}-\norm{z-x}_{g(x)}^{2}}\leq\frac{2\veps r^{2}}{d}}\geq1-\veps\,.
\]
In this condition, the event is understood to include $z\in\inter\K$,
so that $g(z)$ is defined.
\end{defn}

\paragraph{Mixing analysis.}

We use the $\dw$ framework of \cite{KV24ipm}. In that framework,
the local metric must satisfy SSC, LTSC, ASC, and the following symmetry
condition. The symmetry parameter originates in \cite[Definition~2]{LLV20strong};
we call $g$ $\bar{\nu}$-symmetric if 
\begin{equation}
\msf D_{g}(x,1)\subseteq\K\cap(2x-\K)\subseteq\msf D_{g}(x,\sqrt{\bar{\nu}})\quad\text{for all }x\in\inter\K\label{eq:symmetry}
\end{equation}
It holds that $\onu=O(\nu^{2})$ in general.
\begin{lem}[{Mixing of $\dw$ {\cite[Theorem~1]{KV24ipm}}}]
\label{prop:gcdw-framework} Suppose that a metric $g$ is $\bar{\nu}$-symmetric,
SSC, LTSC, and ASC. Let $\pi$ be an exponential distribution on $\K$,
and initialize the $\dw$ at $\pi_{0}$. Then, for $r=\Theta(1)$,
the walk outputs a sample $X\sim\mu$ satisfying $\chi^{2}(\mu\mmid\pi)\leq\veps$
in 
\[
O\bpar{d\bar{\nu}\log\frac{\chi^{2}(\pi_{0}\mmid\pi)}{\veps}}\quad\text{iterations.}
\]
\end{lem}

Here we use the exponential-target specialization of \cite[Theorem~1]{KV24ipm}.
The one-step coupling and isoperimetry argument in \cite{KV24ipm}
gives conductance $\Phi\gtrsim1/\sqrt{d\bar{\nu}}$. Since the $\dw$
is lazy and reversible with respect to $\pi$, Cheeger's inequality
gives an $L^{2}(\pi)$ spectral gap $\gamma\gtrsim\Phi^{2}$. Thus,
for $f_{0}:=\D\pi_{0}/\D\pi$ and the transition kernel $P$, there
exists a universal constant $c>0$ such that
\[
\chi^{2}(\pi_{0}P^{t}\mmid\pi)=\norm{P^{t}(f_{0}-1)}_{L^{2}(\pi)}^{2}\leq(1-\gamma)^{2t}\,\norm{f_{0}-1}_{L^{2}(\pi)}^{2}\leq\exp\bpar{-\frac{ct}{d\onu}}\,\chi^{2}(\pi_{0}\mmid\pi)\,.
\]
Choosing $t\gtrsim d\bar{\nu}\log\frac{\chi^{2}(\pi_{0}\mmid\pi)}{\veps}$
yields the displayed $\chi^{2}$ guarantee. We refer readers to \cite{JC24regularized}
for a simpler mixing analysis.

The final step of the reduction is the following simple scaling observation.

\begin{prop}
\label{lem:scaling-rule} For $\veps>0$, let $g_{0}$ be a metric
satisfying the ASC condition with $r_{\veps}=r(\veps)$ (possibly
dimension-dependent). If $g=Lg_{0}$ with $L\geq1$, then the metric
$g$ satisfies the ASC condition with $r_{\veps}':=L^{1/2}r_{\veps}$.
Moreover, if $g_{0}$ is $\bar{\nu}$-symmetric, then $g$ is $L\bar{\nu}$-symmetric. 
\end{prop}

\begin{proof}
Using the ASC of $g_{0}$, with respect to $z\sim\msf N(x,\frac{r^{2}}{d}\,g(x)^{-1})$,
if $r\leq r_{\veps}$,
\[
\abs{\norm{z-x}_{Lg_{0}(z)}^{2}-\norm{z-x}_{Lg_{0}(x)}^{2}}=\abs{L\,(\norm{z-x}_{g_{0}(z)}^{2}-\norm{z-x}_{g_{0}(x)}^{2})}\leq L\,\frac{2\veps\,(r/L^{1/2})^{2}}{d}=\frac{2\veps r^{2}}{d}
\]
with probability at least $1-\veps$. Hence, if $r\leq r_{\veps}'=L^{1/2}r_{\veps}$,
then the ASC condition for $g=Lg_{0}$ holds. The symmetry statement
follows directly from $\norm{y-x}_{Lg_{0}(x)}^{2}=L\,\norm{y-x}_{g_{0}(x)}^{2}$.
\end{proof}

\subsection{The Lee--Sidford metric for linear inequalities\label{subsec:ls-metric}}

For $p\geq2$, the $\ell_{p}$-Lewis weight of $A_{x}$ \cite{LS19solving}
is the vector $w_{x}\in\R^{m}$ satisfying 
\[
w_{x}=\sigma\bpar{\Diag(w_{x})^{1/2-1/p}A_{x}}\,,
\]
where $\sigma(B):=\Diag(B\,(B^{\T}B)^{-1}B^{\T})$ denotes the leverage
scores of the matrix $B$. Equivalently, for $W_{x}:=\Diag w_{x}$,
\[
W_{x}=\Diag\bpar{W_{x}^{\half-\frac{1}{p}}A_{x}\,(A_{x}^{\T}W_{x}^{1-\frac{2}{p}}A_{x})^{-1}A_{x}^{\T}W_{x}^{\half-\frac{1}{p}}}\,.
\]
Throughout, we take $p=\Theta(\polylog m)$ and suppress polylogarithmic
factors in $m$, $p$, and $\veps^{-1}$.

We define the baseline Lee--Sidford (LS) metric \cite{LLV20strong}
as 
\begin{equation}
g_{0}(x):=A_{x}^{\T}W_{x}^{1-\frac{2}{p}}A_{x}\,.\label{eq:LS-metric}
\end{equation}
The standard LS metric used in the mixing theorem is a fixed polylogarithmic
multiple of $g_{0}$. This scalar is suppressed in $\widetilde{O}(\cdot)$
and does not affect the ASC calculation. It is known from \cite[Lemmas~4.2 and~4.3]{LLV20strong}
that, with this standard polylogarithmic normalization, the LS metric
is SSC and has convex log-determinant. The latter gives 
\[
\tr(g_{0}(x)^{-1}\Dd^{2}g_{0}(x)[h,h])\geq\tr(g_{0}(x)^{-1}\Dd g_{0}(x)[h]g_{0}(x)^{-1}\Dd g_{0}(x)[h])\geq0\,,
\]
which is the LTSC condition. Thus, the main task in this paper is
to prove the ASC estimate for $g_{0}$ at the largest possible radius
before the final scaling.

\subsubsection{Lewis-weight calculus}

We use the following Lewis-weight identities and estimates throughout
the ASC proof. The first lemma provides the basic size controls for
slack directions and Lewis-weight derivatives. Since $0\preceq W_{x}\preceq I$
and $1-2/p\in[0,1]$, we have 
\[
A_{x}^{\T}W_{x}A_{x}\preceq A_{x}^{\T}W_{x}^{1-\frac{2}{p}}A_{x}=g_{0}(x)\,.
\]

\begin{lem}[\cite{LS19solving}]
 \label{lem:usefulFactLewis} Let $W_{x}=\Diag(w_{x}(A_{x}))\in\mbb S_{++}^{m}$
be the $\ell_{p}$-Lewis weights and $g(x)=A_{x}^{\T}W_{x}A_{x}$
the Lewis-weights metric, and $h\in\Rd$. 
\begin{itemize}
\item \textup{(Lemma 26)} $\max_{i\in[m]}\frac{[\sigma(W_{x}^{1/2}A_{x})]_{i}}{(w_{x})_{i}}\leq2m^{\frac{2}{p+2}}$. 
\item \textup{(Lemma 33)} $\norm{A_{x}h}_{W_{x}}=\norm h_{g(x)}$ and $\norm{A_{x}h}_{\infty}\leq\sqrt{2}m^{\frac{1}{p+2}}\,\norm h_{g(x)}$. 
\item \textup{(Lemma 34)} $\norm{W_{x}^{-1}w_{x,h}'}_{W_{x}}\leq p\,\norm h_{g(x)}$. 
\end{itemize}
\end{lem}

The next identity is the main way derivatives of the Lewis weights
enter the proof. Below, $A^{(2)}(=A^{\circ2})$ is the matrix defined
as $(A^{(2)})_{ij}=A_{ij}^{2}$ for a matrix $A$.
\begin{lem}[\cite{LS19solving}, Lemma 24]
\label{lem:DWh} The directional derivative of the $\ell_{p}$-Lewis
weight $W_{x}$ in direction $h\in\Rd$ is 
\[
W_{x,h}':=\Dd W_{x}[h]=-\Diag(W_{x}^{\half}N_{x}W_{x}^{\half}s_{x,h})\,,
\]
where $s_{x,h}:=A_{x}h$, 
\[
P_{x}:=W_{x}^{\half-\frac{1}{p}}A_{x}g_{0}(x)^{-1}A_{x}^{\T}W_{x}^{\half-\frac{1}{p}}\,,\qquad P_{x}^{(2)}:=P_{x}^{\circ2}\,,
\]
$c_{p}:=1-2/p$, $\Lambda_{x}\defeq W_{x}-P_{x}^{(2)}$, $\bar{\Lambda}_{x}\defeq W_{x}^{-\half}\Lambda_{x}W_{x}^{-\half}$,
and $N_{x}\defeq2\bar{\Lambda}_{x}(I-c_{p}\bar{\Lambda}_{x})^{-1}$. 
\end{lem}

The matrices in Lemma~\ref{lem:DWh} satisfy
\begin{equation}
P_{x}^{(2)},\Lambda_{x}\preceq W_{x}\preceq I\,.\label{eq:lewisBasic-PWI}
\end{equation}
We also use the following bounds on $N_{x}$ and its first directional
derivative. 
\begin{lem}[\cite{LS19solving}]
\label{lem:LS-comp-tool} Let $x\in\inter\K$ and $h\in\Rd$. For
$c_{p}=1-2/p$ with $p>2$, let $\bar{\Lambda}_{x}:=W_{x}^{-\half}\Lambda_{x}W_{x}^{-\half}=I-W_{x}^{-\half}P_{x}^{(2)}W_{x}^{-\half}$,
$N_{x}\defeq2\bar{\Lambda}_{x}(I-c_{p}\bar{\Lambda}_{x})^{-1}$ and
$\theta_{x}=A_{x}^{\T}W_{x}A_{x}$. 
\begin{itemize}
\item \textup{(Lemma 31)} $N_{x}$ is symmetric and $0\preceq N_{x}\preceq pI$. 
\item \textup{(Lemma 34)} $\norm{W_{x}^{-1}w_{x,h}'}_{\infty}\leq p(\sqrt{2}m^{\frac{1}{p+2}}+p/2)\,\norm h_{\theta_{x}}$. 
\item \textup{(Lemma 37)} $\norm{(I+N_{x})^{-\half}\Dd N_{x}[h]\,(I+N_{x})^{-\half}}\leq4p^{5/2}\,\norm h_{\theta_{x}}$. 
\end{itemize}
\end{lem}

Finally, the closeness estimate below propagates coordinate bounds
from the base point to nearby points on the proposal segment. 
\begin{lem}[\cite{LS19solving}, Lemma 35]
\label{lem:weight-close} In the same setting as above, let $x_{t}=x+t\,h$,
$\sigma_{t,i}:=a_{i}^{\T}x_{t}-b_{i}$, and $w_{t}=w_{x_{t}}$. Define
$z_{t,\alpha}\in\R^{m}$ by $[z_{t,\alpha}]_{i}:=\frac{\D}{\D t}\log\frac{[w_{t,i}]^{\alpha}}{\sigma_{t,i}}$.
Then, 
\[
\norm{z_{t,\alpha}}_{\infty}\leq\bpar{\sqrt{2}(1+|\alpha|p)m^{\frac{1}{p+2}}+p\,|\alpha|\,\max(1,p/2)}\,\norm h_{A_{x_{t}}^{\T}W_{x_{t}}A_{x_{t}}}\,.
\]
\end{lem}

\section{Average self-concordance for the LS metric\label{sec:ls-asc}}

The main estimate is the following ASC bound for the LS metric $g_{0}$.
\begin{prop}
\label{prop:baseline-asc} Given $\veps>0$, choose $R_{\veps}=\widetilde{\Theta}(d^{-1/8})>0$
so that 
\[
R_{\veps}+R_{\veps}^{2}+R_{\veps}^{4}d^{1/2}\leq\frac{\veps}{\polylog(m/\veps)}\,.
\]
Then, for any $x\in\inter\K$, the following ASC estimate holds whenever
$0<r\leq R_{\veps}$: 
\[
\P_{z\sim\msf N_{g_{0}}^{r}(x)}\bpar{\abs{\norm{z-x}_{g_{0}(z)}^{2}-\norm{z-x}_{g_{0}(x)}^{2}}\leq\frac{2\veps r^{2}}{d}}\geq1-\veps\,.
\]
\end{prop}

\subsection{Proof sketch}

The LS metric and ASC quantity are affine-invariant, so we may fix
$x=0$ and $g_{0}(x)=I_{d}$. Let $F(t):=h^{\T}g_{0}(x+th)h$ for
$0\leq t\leq\eta:=r/d^{1/2}$ and $h\sim\msf N(0,I_{d})$ . Then,
the ASC estimate follows from $\abs{F(\eta)-F(0)}\leq2\veps$ with
high probability. As in previous works, we expand $F(t)$ at $t=0$
and present high-probability bounds. To explain the proof idea, we
describe approaches in prior works and then introduce notation used
throughout in the computation.

\paragraph{Prior approaches.}

Earlier works expand $F(t)$ to second-order derivatives, and bound
$\abs{F'(0)}\lesssim d^{1/2}$ and $\abs{F''(t)}\lesssim d^{3/2}$
(w.h.p.), respectively. Since $\eta=r/d^{1/2}$, the second-order
term $\eta^{2}F''$ is of order $d^{1/2}$. This $d^{1/2}$-factor
forces previous analyses to scale up their base Lewis metric by $d^{1/2}$,
thereby incurring the previously best $d^{2.5}$-mixing of the $\dw$.

A natural idea is to expand the second-order term further and to attempt
to bound $\abs{F''(0)}\lesssim d$. The trouble is that $F''(0)$
already has $6$ distinct terms, and bounding some of them requires
intensive $L^{2}$-norm estimations of Gaussian polynomials. Even
worse, additional expansion requires bounds on $F'''(t)$, but the
number of distinct terms in $F'''$ dramatically increases, and Lewis-weight
calculus was established up to second order in \cite{LS19solving}.
Taken together, these obstacles make it daunting to improve the mixing
guarantee of the $\dw$ with Lewis weights, which has remained at
$\Otilde(d^{2.5})$ since \cite{CDW18MCMC,KV24ipm,JC24regularized}.

\paragraph{Refined approach.}

Define $\alpha:=1-\frac{2}{p}$ and $\beta:=\frac{\alpha}{2}=\half-\frac{1}{p}$.
Along $x_{t}:=x+th$, define $A_{t}:=A_{x_{t}}$, $W_{t}:=W_{x_{t}}$,
$N_{t}:=N_{x_{t}}$, and 
\[
s_{t}:=A_{t}h\,,\qquad u_{t}:=W_{t}^{\beta}s_{t}\,,\qquad v_{t}:=W_{t}^{1/2}s_{t}\,,\qquad\rho_{t}:=W_{t}^{-1/p}s_{t}\,,\qquad q_{t}:=W_{t}^{-1/2}u_{t}^{\circ2}\,.
\]
By Lemma~\ref{lem:DWh}, $W_{t}'=-\Diag(W_{t}^{1/2}N_{t}W_{t}^{1/2}s_{t})$.
For $\ell_{t}:=\log w_{t}$ (elementwise), we have $\ell_{t}'=-W_{t}^{-1/2}N_{t}v_{t}$.
We also define $z_{t}:=\beta\ell_{t}'-s_{t}$ and $Z_{t}:=\Diag z_{t}$.

With the notation above, we have $F(t)=\norm{u_{t}}^{2}$. Since $s_{t}'=-s_{t}^{\circ2}$,
we have $u_{t}'=Z_{t}u_{t}$. Therefore, 
\[
F'(t)=2\inner{u_{t},Z_{t}u_{t}}=2\beta\inner{u_{t}^{\circ2},\ell_{t}'}-2\inner{u_{t}^{\circ2},s_{t}}=-2\beta q_{t}^{\T}N_{t}v_{t}-2\inner{s_{t},W_{t}^{\alpha}s_{t}^{\circ2}}=:-2\beta H_{1}(t)-2E_{1}(t)\,.
\]
Unlike earlier works, we group terms in a better way, and isolate
and expand only ``bottleneck'' terms in $F'(t)$ that are not controlled
directly on a high-probability event. Since $F(\eta)-F(0)=\int_{0}^{\eta}F'(t)\,\D t$,
we may want to bound $H_{1}(t)$ and $E_{1}(t)$ by $O(\eta^{-1})=O(d^{1/2})$(w.h.p.),
but, as shown below, the argument can only afford $O(d)$ bounds,
so in the first-order term, both $E_{1}(t)$ and $H_{1}(t)$ are problematic.
This leads us to differentiate both terms further. 

To present our approach in an organized way, let us define
\[
H_{k}(t):=q_{t}^{\T}N_{t}^{(k-1)}v_{t}\quad\text{for }k\geq1\,,
\]
where $N_{t}^{(j)}$ denotes the $j$-th derivative along the path
$x_{t}=x+th$. Then, we can write 
\begin{equation}
H_{k}'(t)=G_{k+1}(t)+H_{k+1}(t)\,,\qquad G_{k+1}(t):=(q_{t}')^{\T}N_{t}^{(k-1)}v_{t}+q_{t}^{\T}N_{t}^{(k-1)}v_{t}'\,.\label{eq:Gk-defn}
\end{equation}

Additional expansion produces base-point Gaussian polynomials $\eta\,(E_{1}(0)+H_{1}(0))$
along with pathwise terms of order $O(\eta^{2}\,\{E_{1}'(t)+H_{1}'(t)\})$.
In contrast to $E_{1}(t)$ and $H_{1}(t)$, the base-point terms afford
an additional $d^{1/2}$-reduction in its order through the following
standard concentration estimate for Gaussian polynomials. 
\begin{lem}[Concentration of Gaussian polynomials]
\label{lem:conc-gaussian-poly} For $d\geq1$, let $P:\Rd\to\R$
be a polynomial of degree $n$. For any $t\geq(2e)^{n/2}$, 
\[
\P_{h\sim\msf N(0,I_{d})}\bpar{|P(h)|\geq t\sqrt{\E[P(h)^{2}]}}\leq\exp\bpar{-\frac{nt^{2/n}}{2e}}\,.
\]
\end{lem}

Namely, $E_{1}(0)$ and $H_{1}(0)$ are bounded by their $L^{2}$-norm
with high probability. As seen shortly, their $L^{2}$-norms are $O(d^{1/2})$,
so the first-order term $\eta\,(E_{1}(0)+H_{1}(0))$ is bounded as
$O(1)$. We next consider the second-order term $\eta^{2}\,(E_{1}'+H_{1}')=\eta^{2}\,(E_{1}'+G_{2}+H_{2})$.
Fortunately, $E_{1}'(t)$ and $G_{2}(t)$ are $O(d)=O(\eta^{-2})$
in a good high-probability region, so these terms can be handled directly.
However, the bottleneck in prior analyses is that $H_{2}(t)$ was
only bounded by $O(d^{3/2})$. This extra $d^{1/2}$-factor forces
prior analyses to scale up their base Lewis metric by $d^{1/2}$,
thereby incurring the previously best $d^{2.5}$-mixing of the $\dw$.

We find that $H_{k}(t)$ is indeed a bottleneck term in our argument,
so expand only this term up to fourth-order. Similar to the first
and second-order terms, the base-point Gaussian polynomial $H_{k+1}(0)$
is controlled in $L^{2}$ and bounded by Lemma~\ref{lem:conc-gaussian-poly},
and the pathwise term $G_{k+1}$ and the terminal term $H_{4}$ are
bounded deterministically on the good event in Lemma~\ref{lem:good-term-bounds}:

\begin{align}
\abs{F(\eta)-F(0)} & \lesssim\eta\,\abs{E_{1}(0)+\beta H_{1}(0)}+\eta^{2}\bpar{\sup_{t\leq\eta}\abs{E_{1}'(t)}+\sup_{t\leq\eta}\abs{G_{2}(t)}+\abs{H_{2}(0)}}\nonumber \\
 & \qquad+\eta^{3}\bpar{\sup_{t\leq\eta}\abs{G_{3}(t)}+\abs{H_{3}(0)}}+\eta^{4}\sup_{t\leq\eta}\bpar{\abs{G_{4}(t)}+\abs{H_{4}(t)}}\,.\label{eq:hard-recursion}
\end{align}

The rest of the proof has two independent tasks: (1) pathwise bounds
on the good event, and (2) $L^{2}$-norm estimates of the base-point
terms $E_{1}(0)+\beta H_{1}(0)$, $H_{2}(0)$, and $H_{3}(0)$. For
the first task, we establish higher-order Lewis-weight calculus. As
seen later in Lemma~\ref{lem:good-term-bounds}, isolation of $G_{k}$
and $H_{k}$ in $H_{k-1}'$ proves fruitful, since a high-probability
bound for $G_{k}$ follows \emph{inductively}, rather than via ad
hoc computation.

Regarding the second task, $L^{2}$-norm estimation of Gaussian polynomials
$F'(0)$ was arguably the most technical part in prior works. We streamline
this using the Gaussian Poincar\'e inequality. As for $H_{2}(0)$
and $H_{3}(0)$, these are Gaussian polynomials of degree $4$ and
$5$, so we would need to compute the expectation of Gaussian polynomials
of degree $8$ and $10$ for their $L^{2}$-norms. As far as we can
tell, these are technically the most difficult parts and require principled
approaches to go beyond the previous work.

To this end, we work with a decomposition of $H_{2}(0),H_{3}(0)$
into \emph{Hermite polynomials}. Since Hermite polynomials are orthogonal
in $L^{2}(\gamma)$, their $L^{2}$-norm estimation can be simplified.
Furthermore, this toolkit can be handled in a principled way via a
\emph{multiple stochastic integral} (MSI) (see \S\ref{subsec:MSI}),
combined with Wiener chaos decomposition. Known results on MSI substantially
streamline $L^{2}$-norm estimation, which were previously almost
infeasible without it.

\subsection{Pathwise bound}

We establish higher-order calculus of Lewis weights as follows. 
\begin{lem}[Good event]
\label{lem:good-event} In the affine normalization $x=0$ and $g_{0}(x)=I_{d}$,
let $h\sim\msf N(0,I_{d})$, $x_{t}=x+th$, $0<r\leq R_{\veps}=\widetilde{\Theta}(d^{-1/8})$,
and $\eta=r/\sqrt{d}$. There exists a good event $\mc E$ of $\P(\mc E)\geq1-\veps/20$
on which for all $t\in[0,\eta]$,
\begin{align*}
\norm{\rho_{t}}_{\infty},\norm{s_{t}}_{\infty} & \lesssim1\\
\norm{u_{t}},\norm{v_{t}},\norm{N_{t}v_{t}},\norm{q_{t}},\norm{q_{t}'},\norm{v_{t}'},\norm{z_{t}}_{\infty},\norm{z_{t}}_{W_{t}},\norm{N_{t}'} & \lesssim d^{1/2}\\
\norm{z_{t}'}_{W_{t}},\norm{v_{t}''},\norm{N_{t}''} & \lesssim d\\
\norm{z_{t}''}_{W_{t}},\norm{N_{t}'''} & \lesssim d^{3/2}\,.
\end{align*}
\end{lem}

\begin{proof}
All estimates below are uniform for $t\in[0,\eta]$. Note that $\rho_{0}=W_{0}^{-1/2}u_{0}$,
and define $U:=W_{0}^{\beta}A_{0}$. Then, $U^{\T}U=I_{d}$ and $UU^{\T}=P$
with $P_{ii}=w_{i}$. Hence, 
\[
\rho_{0,i}=w_{i}^{-1/2}(Uh)_{i}\sim\msf N(0,1)\,.
\]
Using the standard bound $\norm h\lesssim d^{1/2}$ and a union bound,
we have that with probability at least $1-\veps/20$, 
\[
\norm h\lesssim d^{1/2}\,,\qquad\norm{\rho_{0}}_{\infty}\lesssim1
\]
after the hidden logarithmic factor in $\lesssim$ is chosen as a
function of $m$ and $\veps$. Hereafter, we always work on this event.
For $L\asymp d^{1/4}\polylog(md)$, we note that $\norm{x_{t}-x}_{g}=\norm{x_{t}-x}_{Lg_{0}}\leq rL^{1/2}\cdot\norm h/d^{1/2}$,
so by reducing a constant scale of $r$, we can enforce $\norm{x_{t}-x}_{g}\leq1$
on this event. By the $\onu$-symmetry \eqref{eq:symmetry}, the Dikin
ellipsoid of radius $1$ with respect to $g$ is contained $\K$,
so the path is contained in $\K$. 

By direct computation, for $\sigma_{t,i}:=a_{i}^{\T}x_{t}-b_{i}$,
\begin{equation}
\frac{\D}{\D t}\log\sigma_{t,i}^{-1}=-s_{t,i}\,,\qquad\frac{\D}{\D t}\log\abs{\rho_{t,i}}=-\frac{1}{p}\,\ell_{t,i}'-s_{t,i}\,.\label{eq:sigma-rho-derv}
\end{equation}
This derivative is bounded by $\norm h_{A_{t}^{\T}W_{t}A_{t}}$ (up
to logarithmic factors) by Lemma~\ref{lem:weight-close}, applied
with the exponents $0$ and $-1/p$. Since $\norm{x_{t}-x}_{g_{0}(x)}=t\norm h\leq\eta\norm h\lesssim r\,\polylog(m/\veps)$,
the choice $r\leq R_{\veps}$ makes the last quantity a sufficiently
small absolute constant. By the self-concordance of the standard LS
metric (after absorbing some polylogarithmic scalar relating it to
$g_{0}$ into the suppressed logarithmic factors), uniformly for $t\in[0,\eta]$,
\[
\norm h_{g_{0}(x_{t})}\lesssim\norm h_{g_{0}(x)}=\norm h\lesssim d^{1/2}\,.
\]
Since $A_{t}^{\T}W_{t}A_{t}\preceq g_{0}(x_{t})$, the RHS in Lemma~\ref{lem:weight-close}
is a polylogarithmic multiple of $\norm h_{g_{0}(x_{t})}$, hence
$\widetilde{O}(d^{1/2})$ on the whole interval $[0,\eta]$. Thus,
uniformly over this interval, $\norm{\rho_{t}}_{\infty}\lesssim1$
and $\norm{s_{t}}_{\infty}\lesssim1$. Moreover, as $z_{t}=\beta\ell_{t}'-s_{t}$,
the derivative of $\log\rho_{t,i}$ above implies that $\norm{\ell_{t}'}_{\infty},\norm{z_{t}}_{\infty}\lesssim d^{1/2}$.

Using these stability results,
\[
\norm{u_{t}}^{2}=\sum_{i}w_{t,i}\rho_{t,i}^{2}\lesssim\sum_{i}w_{t,i}=d\,,\qquad\norm{v_{t}}=\norm{W_{t}^{1/2}s_{t}}=\norm{W_{t}^{1/p}u_{t}}\leq\norm{u_{t}}\lesssim d^{1/2}\,.
\]
By Lemma~\ref{lem:LS-comp-tool}, $\norm{N_{t}}\lesssim1$ up to
polylogarithmic factors, so $\norm{N_{t}v_{t}}\lesssim d^{1/2}$.
Also, 
\begin{equation}
\norm{q_{t}}^{2}=\sum_{i}w_{t,i}\rho_{t,i}^{4}\leq\norm{\rho_{t}}_{\infty}^{4}d\lesssim d\,,\qquad\norm{z_{t}}_{W_{t}}=\norm{\beta W_{t}^{1/2}\ell_{t}'-W_{t}^{1/2}s_{t}}=\norm{-\beta N_{t}v_{t}-v_{t}}\lesssim d^{1/2}\,.\label{eq:qt-ztWt}
\end{equation}

Next, since 
\begin{equation}
v_{t}'=(\half\,\ell_{t}'-s_{t})\circ v_{t}=-s_{t}\circ v_{t}-\half\,(N_{t}v_{t})\circ s_{t}\,,\label{eq:vt'}
\end{equation}
we have 
\[
\norm{v_{t}'}\lesssim\norm{s_{t}}_{\infty}\,(\norm{v_{t}}+\norm{N_{t}v_{t}})\lesssim d^{1/2}\,.
\]
In $W_{t}^{1/2}z_{t}'=-\frac{\beta}{2}\,W_{t}^{1/2}(\ell_{t}')^{\circ2}-\beta N_{t}'v_{t}-\beta N_{t}v_{t}'+W_{t}^{1/2}s_{t}^{\circ2}$,
the first term is bounded as 
\begin{equation}
\norm{W_{t}^{1/2}(\ell_{t}')^{\circ2}}\leq\norm{\ell_{t}'}_{\infty}\,\norm{W_{t}^{1/2}\ell_{t}'}=\norm{\ell_{t}'}_{\infty}\,\norm{N_{t}v_{t}}\lesssim d\,.\label{eq:Wtellt'2}
\end{equation}
The $N_{t}'v_{t}$ term is $O(d)$, since $\norm{N_{t}'}\lesssim d^{1/2}$
due to Lemma~\ref{lem:LS-comp-tool}-3 and self-concordance of $g_{0}$.
The last two terms are $O(d^{1/2})$ in norm. Thus, $\norm{z_{t}'}_{W_{t}}\lesssim d$.

Differentiating $q_{t}=W_{t}^{-1/2}u_{t}^{\circ2}$ and using $u_{t}'=Z_{t}u_{t}=(\beta\ell_{t}'-s_{t})\circ u_{t}$
\[
q_{t}'=-\half\,\ell_{t}'\circ q_{t}+2W_{t}^{-1/2}(u_{t}\circ u_{t}')=-\half\,\ell_{t}'\circ q_{t}+2\beta\,(\ell_{t}'\circ q_{t})-2\,(s_{t}\circ q_{t})=\bpar{\half-\frac{2}{p}}\,\ell_{t}'\circ q_{t}-2\,(s_{t}\circ q_{t})\,.
\]
Since $\ell_{t}'\circ q_{t}=-(N_{t}v_{t})\circ\rho_{t}^{\circ2}$,
\[
\norm{q_{t}'}\lesssim\norm{s_{t}}_{\infty}\,\norm{q_{t}}+\norm{\rho_{t}}_{\infty}^{2}\,\norm{N_{t}v_{t}}\lesssim d^{1/2}\,.
\]

For $N_{t}$, let $\widehat{B}_{t}=W_{t}^{\beta}A_{t}$, choose a
smooth orthonormal basis $U_{t}$ of $\cols\widehat{B}_{t}$, and
define $P_{t}:=U_{t}U_{t}^{\T}$ and $P_{t}^{\perp}:=I-P_{t}$. For
$P_{t}^{(2)}:=P_{t}^{\circ2}$, recall 
\[
N_{t}=2\bar{\Lambda}_{t}(I-c_{p}\bar{\Lambda}_{t})^{-1}\,,\qquad\bar{\Lambda}_{t}=I-W_{t}^{-1/2}P_{t}^{(2)}W_{t}^{-1/2}\,.
\]
Denote the $i$-th row of $U_{t}$ by $u_{t,i}$, and define the $i$-th
row of $\Phi_{t}$ by $\phi_{t,i}:=w_{t,i}^{-1/2}(u_{t,i}\otimes u_{t,i})$.
Then $\Phi_{t}\Phi_{t}^{\T}=W_{t}^{-1/2}P_{t}^{(2)}W_{t}^{-1/2}$,
so $\bar{\Lambda}_{t}=I-\Phi_{t}\Phi_{t}^{\T}$. By \eqref{eq:lewisBasic-PWI},
$\norm{\Phi_{t}}\leq1$.

Writing $R_{t}:=(I-c_{p}\bar{\Lambda}_{t})^{-1}$, we have 
\begin{equation}
N_{t}'=2R_{t}\bar{\Lambda}_{t}'R_{t}\,,\qquad N_{t}''=2R_{t}\bar{\Lambda}_{t}''R_{t}+4c_{p}R_{t}\bar{\Lambda}_{t}'R_{t}\bar{\Lambda}_{t}'R_{t}\,.\label{eq:N-derivatives}
\end{equation}
Thus, $\norm{N_{t}'}\lesssim\norm{\bar{\Lambda}_{t}'}$ and $\norm{N_{t}''}\lesssim\norm{\bar{\Lambda}_{t}'}^{2}+\norm{\bar{\Lambda}_{t}''}$.
Also, since $\bar{\Lambda}_{t}=I-\Phi_{t}\Phi_{t}^{\T}$, it holds
that $\norm{\bar{\Lambda}_{t}'}\lesssim\norm{\Phi_{t}'}$ and $\norm{\bar{\Lambda}_{t}''}\lesssim\norm{\Phi_{t}''}+\norm{\Phi_{t}'}^{2}$.
Moreover, the chain rule yields
\begin{equation}
\phi_{i}'(t)=\Dd R(u_{t,i})[u_{t,i}']\,,\qquad\phi_{i}''(t)=\Dd^{2}R(u_{t,i})[u_{t,i}',u_{t,i}']+\Dd R(u_{t,i})[u_{t,i}'']\,.\label{eq:phi-derivatives}
\end{equation}
By Lemma~\ref{lem:row-map}, applied row by row to $\phi_{t,i}=R(u_{t,i})=R(u_{i}(t))$,
\begin{equation}
\norm{\Phi_{t}'}_{\frob}\lesssim\norm{U_{t}'}_{\frob}\,,\qquad\norm{\Phi_{t}''}_{\frob}^{2}\lesssim\norm{U_{t}''}_{\frob}^{2}+\sum_{i}\frac{\norm{u_{t,i}'}^{4}}{w_{t,i}}\,.\label{eq:Phi-derivative-bounds}
\end{equation}

We can write $\widehat{B}_{t}=U_{t}M_{t}$ for invertible $M_{t}$,
so equating $\widehat{B}_{t}'=Z_{t}\widehat{B}_{t}$ (due to $u_{t}'=Z_{t}u_{t}$)
and $\widehat{B}_{t}'=U_{t}'M_{t}+U_{t}M_{t}'$ leads to $U_{t}'=Z_{t}U_{t}-U_{t}M_{t}'M_{t}^{-1}$.
Hence,
\[
U_{t}'=P_{t}^{\perp}U_{t}'+P_{t}U_{t}'=P_{t}^{\perp}Z_{t}U_{t}+U_{t}U_{t}^{\T}U_{t}'\,.
\]
By Lemma~\ref{lem:path-normal-frame}, we can replace $U_{t}$ by
a ``good'' orthonormal basis for $\cols\widehat{B}_{t}$ such that
$U_{t}^{\T}U_{t}'=0$, and this leads to $U_{t}'=P_{t}^{\perp}Z_{t}U_{t}$.
Differentiating this identity gives
\begin{equation}
U_{t}'=P_{t}^{\perp}Z_{t}U_{t}\,,\qquad U_{t}''=-P_{t}'Z_{t}U_{t}+P_{t}^{\perp}Z_{t}'U_{t}+P_{t}^{\perp}Z_{t}U_{t}'\,.\label{eq:Ut-derivatives}
\end{equation}
Using $Z_{t}=\Diag z_{t}$ and $x^{\T}(A\circ B)y=\tr(\Diag(x)A\Diag(y)B^{\T})$,
\[
\norm{U_{t}'}_{\frob}^{2}=z_{t}^{\T}(P_{t}\circ P_{t}^{\perp})z_{t}=z_{t}^{\T}(W_{t}-P_{t}^{\circ2})z_{t}\leq\norm{z_{t}}_{W_{t}}^{2}\lesssim d\,.
\]
For the second derivative in \eqref{eq:Ut-derivatives}, note that
by $P_{t}'=U_{t}'U_{t}^{\T}+U_{t}(U_{t}')^{\T}$ and $\norm{U_{t}}\leq1$,
we have $\norm{P_{t}'}\leq2\norm{U_{t}'}$. Also, $\norm{Z_{t}}=\norm{z_{t}}_{\infty}$.
Thus,
\begin{align}
\norm{P_{t}'Z_{t}U_{t}}_{\frob} & \leq\norm{P_{t}'}\,\norm{Z_{t}U_{t}}_{\frob}=\norm{P_{t}'}\,\norm{z_{t}}_{W_{t}}\lesssim d^{1/2}\cdot d^{1/2}=d\,,\nonumber \\
\norm{P_{t}^{\perp}Z_{t}'U_{t}}_{\frob} & \leq\norm{Z_{t}'U_{t}}_{\frob}=\norm{z_{t}'}_{W_{t}}\lesssim d\,,\label{eq:ptperp-zt-ut}\\
\norm{P_{t}^{\perp}Z_{t}U_{t}'}_{\frob} & \leq\norm{P_{t}^{\perp}}\,\norm{Z_{t}}\,\norm{U_{t}'}_{\frob}\lesssim d\,.\nonumber 
\end{align}
Hence, $\norm{U_{t}''}_{\frob}\lesssim d$.

Finally, from $u_{t,i}'=e_{i}^{\T}(I-P_{t})Z_{t}U_{t}=e_{i}^{\T}Z_{t}U_{t}-e_{i}^{\T}P_{t}Z_{t}U_{t}$
and $\norm{u_{t,i}}^{2}=w_{i}$, we obtain that $\norm{e_{i}^{\T}Z_{t}U_{t}}\lesssim\norm{z_{t}}_{\infty}w_{t,i}^{1/2}$.
Also, $\norm{e_{i}^{\T}P_{t}Z_{t}U_{t}}^{2}=e_{i}^{\T}P_{t}Z_{t}P_{t}Z_{t}P_{t}e_{i}\leq\tr(E_{ii}P_{t}Z_{t}^{2}P_{t})=e_{i}^{\T}P_{t}^{\circ2}z_{t}^{2}\lesssim\norm{z_{t}}_{\infty}^{2}w_{t,i}$.
Thus, 
\begin{equation}
\norm{u_{t,i}'}\lesssim\norm{z_{t}}_{\infty}\,w_{t,i}^{1/2}\lesssim d^{1/2}w_{t,i}^{1/2}\,.\label{eq:u'ti-bound}
\end{equation}
Combining this with $\norm{U_{t}'}_{\frob}^{2}\lesssim d$ gives
\[
\sum_{i}\frac{\norm{u_{t,i}'}^{4}}{w_{t,i}}\lesssim d\,\sum_{i}\norm{u_{t,i}'}^{2}=d\,\norm{U_{t}'}_{\frob}^{2}\lesssim d^{2}\,.
\]
By \eqref{eq:Phi-derivative-bounds}, $\norm{\Phi_{t}'}_{\frob}\lesssim d^{1/2}$
and $\norm{\Phi_{t}''}_{\frob}\lesssim d$. The derivative identities
\eqref{eq:N-derivatives} therefore give $\norm{N_{t}'}\lesssim d^{1/2}$
and $\norm{N_{t}''}\lesssim d$.

For the third-order estimates, differentiating $W_{t}^{1/2}z_{t}'=-\frac{\beta}{2}\,W_{t}^{1/2}(\ell_{t}')^{\circ2}-\beta N_{t}'v_{t}-\beta N_{t}v_{t}'+W_{t}^{1/2}s_{t}^{\circ2}$,
\begin{align*}
W_{t}^{1/2}z_{t}'' & =-\beta W_{t}^{1/2}(\ell_{t}'\circ\ell_{t}'')-\frac{\beta}{4}W_{t}^{1/2}(\ell_{t}')^{\circ3}-\beta N_{t}''v_{t}-2\beta N_{t}'v_{t}'-\beta N_{t}v_{t}''\\
 & \qquad+\frac{1}{2}\,\ell_{t}'\circ W_{t}^{1/2}s_{t}^{\circ2}+2W_{t}^{1/2}(s_{t}\circ s_{t}')-\frac{1}{2}\,\ell_{t}'\circ W_{t}^{1/2}z_{t}'\,.
\end{align*}
Since $\ell_{t}'=-W_{t}^{-1/2}N_{t}v_{t}$, we have $W_{t}^{1/2}\ell_{t}''=\half\,\ell_{t}'\circ N_{t}v_{t}-N_{t}'v_{t}-N_{t}v_{t}'$,
so $\norm{W_{t}^{1/2}\ell_{t}''}\lesssim d$. Hence, the first term
is bounded by $\Otilde(d^{3/2})$. Using $\norm{\ell_{t}'}_{\infty},\norm{v_{t}}\lesssim d^{1/2}$
and $\norm{N_{t}''},\norm{z_{t}'}_{W_{t}}\lesssim d$ with \eqref{eq:Wtellt'2},
we can bound the norms of $N_{t}''v_{t}$, $W_{t}^{1/2}(\ell_{t}'\circ\ell_{t}'')$,
$W_{t}^{1/2}(\ell_{t}')^{\circ3}$, and $\ell_{t}'\circ W_{t}^{1/2}z_{t}'$
by $d^{3/2}$. As for $v_{t}''$, differentiating $v_{t}'$ in \eqref{eq:vt'}
gives 
\[
v_{t}''=s_{t}^{\circ2}\circ\bpar{v_{t}+\half\,N_{t}v_{t}}-s_{t}\circ\bpar{v_{t}'+\half\,N_{t}'v_{t}+\half\,N_{t}v_{t}'}\,.
\]
Using $\norm{s_{t}}_{\infty},\norm{N_{t}}\lesssim1$ and $\norm{v_{t}},\norm{v_{t}'},\norm{N_{t}'}\lesssim d^{1/2}$,
we have $\norm{v_{t}''}\lesssim d$, and $\norm{N_{t}v_{t}''}\lesssim d$.
Lastly, by $s_{t}'=-s_{t}^{\circ2}$ and $\norm{s_{t}}_{\infty}\lesssim1$,
the remaining terms can be bounded as $O(d)$. Hence, $\norm{z_{t}''}_{W_{t}}\lesssim d^{3/2}$. 

We now bound $\norm{N_{t}'''}$. Recall that $P_{t}=U_{t}U_{t}^{\T}$,
$P_{t}'=U_{t}'U_{t}^{\T}+U_{t}(U_{t}')^{\T}$, and $P_{t}''=U_{t}''U_{t}^{\T}+2U_{t}'(U_{t}')^{\T}+U_{t}(U_{t}'')^{\T}$,
so $\norm{P_{t}'}_{\frob}\lesssim d^{1/2}$ and $\norm{P_{t}''}_{\frob}\lesssim d$.
Then, from \eqref{eq:Ut-derivatives}, 
\[
U_{t}'''=-P_{t}''Z_{t}U_{t}-2P_{t}'Z_{t}'U_{t}-2P_{t}'Z_{t}U_{t}'+P_{t}^{\perp}Z_{t}''U_{t}+2P_{t}^{\perp}Z_{t}'U_{t}'+P_{t}^{\perp}Z_{t}U_{t}''\,.
\]
Note the identities 
\[
\norm{Z_{t}U_{t}}_{\frob}=\norm{z_{t}}_{W_{t}}\lesssim d^{1/2}\,,\qquad\norm{Z_{t}'U_{t}}_{\frob}=\norm{z_{t}'}_{W_{t}}\lesssim d\,,\qquad\norm{Z_{t}U_{t}'}_{\frob}\leq\norm{z_{t}}_{\infty}\,\norm{U_{t}'}_{\frob}\lesssim d\,,
\]
and recall $\norm{u_{t,i}'}^{2}\lesssim d\,w_{t,i}$. Using this,
\[
\norm{Z_{t}'U_{t}'}_{\frob}^{2}=\sum_{i}(z_{t,i}')^{2}\norm{u_{t,i}'}^{2}\lesssim d\sum_{i}w_{t,i}(z_{t,i}')^{2}=d\,\norm{z_{t}'}_{W_{t}}^{2}\lesssim d^{3}\,.
\]
Using these bounds together with $\norm{P_{t}'}_{\frob}\lesssim d^{1/2}$,
$\norm{P_{t}''}_{\frob}\lesssim d$, $\norm{z_{t}}_{\infty}\lesssim d^{1/2}$,
and $\norm{U_{t}''}_{\frob}\lesssim d$, we obtain 
\begin{align*}
\norm{P_{t}''Z_{t}U_{t}}_{\frob} & \leq\norm{P_{t}''}\,\norm{Z_{t}U_{t}}_{\frob}\lesssim d^{3/2}\,,\qquad\norm{P_{t}'Z_{t}'U_{t}}_{\frob}\leq\norm{P_{t}'}\,\norm{Z_{t}'U_{t}}_{\frob}\lesssim d^{3/2}\,,\\
\norm{P_{t}'Z_{t}U_{t}'}_{\frob} & \leq\norm{P_{t}'}\,\norm{Z_{t}U_{t}'}_{\frob}\lesssim d^{3/2}\,,\qquad\norm{P_{t}^{\perp}Z_{t}''U_{t}}_{\frob}\leq\norm{Z_{t}''U_{t}}_{\frob}=\norm{z_{t}''}_{W_{t}}\lesssim d^{3/2}\,,\\
\norm{P_{t}^{\perp}Z_{t}'U_{t}'}_{\frob} & \leq\norm{Z_{t}'U_{t}'}_{\frob}\lesssim d^{3/2}\,,\qquad\norm{P_{t}^{\perp}Z_{t}U_{t}''}_{\frob}\leq\norm{z_{t}}_{\infty}\,\norm{U_{t}''}_{\frob}\lesssim d^{3/2}\,.
\end{align*}
Therefore, $\norm{U_{t}'''}_{\frob}\lesssim d^{3/2}$.

Differentiating \eqref{eq:phi-derivatives},
\[
\phi_{i}'''(t)=\Dd^{3}R(u_{t,i})[u_{t,i}',u_{t,i}',u_{t,i}']+3\Dd^{2}R(u_{t,i})[u_{t,i}'',u_{t,i}']+\Dd R(u_{t,i})[u_{t,i}''']\,.
\]
Using Lemma~\ref{lem:row-map},$\norm{u_{t,i}'}^{2}\lesssim d\,w_{t,i}$,
$\sum_{i}\norm{u_{t,i}'}^{2}=\norm{U_{t}'}_{\frob}^{2}\lesssim d$,
and $\sum_{i}\norm{u_{t,i}''}^{2}=\norm{U_{t}''}_{\frob}^{2}\lesssim d^{2}$,
\[
\norm{\Phi_{t}'''}_{\frob}^{2}\lesssim\norm{U_{t}'''}_{\frob}^{2}+\sum_{i}\frac{\norm{u_{t,i}'}^{6}}{w_{t,i}^{2}}+\sum_{i}\frac{\norm{u_{t,i}'}^{2}\norm{u_{t,i}''}^{2}}{w_{t,i}}\lesssim d^{3}\,.
\]
Finally, note that $\norm{\bar{\Lambda}_{t}'''}\lesssim\norm{\Phi_{t}'''}_{\frob}+\norm{\Phi_{t}''}_{\frob}\norm{\Phi_{t}'}_{\frob}\lesssim d^{3/2}$.
Differentiating \eqref{eq:N-derivatives},
\[
\norm{N_{t}'''}\lesssim\norm{\bar{\Lambda}_{t}'''}+\norm{\bar{\Lambda}_{t}'}\,\norm{\bar{\Lambda}_{t}''}+\norm{\bar{\Lambda}_{t}'}^{3}\lesssim d^{3/2}\,.
\]
This completes the proof.
\end{proof}
Using this lemma, we present pathwise bounds on the good event.
\begin{lem}[Pathwise bounds for good terms]
\label{lem:good-term-bounds} On the good event $\mc E$, 
\[
\sup_{t\leq\eta}\abs{E_{1}'(t)},\,\sup_{t\leq\eta}\abs{G_{2}(t)}\lesssim d\,,\qquad\sup_{t\leq\eta}\abs{G_{3}(t)}\lesssim d^{3/2}\,,\qquad\sup_{t\leq\eta}\abs{G_{4}(t)}\lesssim d^{2}\,,\qquad\sup_{t\leq\eta}\abs{H_{4}(t)}\lesssim d^{5/2}\,.
\]
\end{lem}

The estimate of $\abs{H_{4}(t)}$ is the source of the new radius
(but also a bottleneck for $\Otilde(d^{2})$-mixing of the $\dw$).
In \eqref{eq:hard-recursion}, it contributes $\eta^{4}d^{5/2}=r^{4}d^{1/2}$,
which leads us to take $R_{\veps}=\widetilde{\Theta}(d^{-1/8})$.
\begin{proof}
For $G_{k+1}(t)=(q_{t}')^{\T}N_{t}^{(k-1)}v_{t}+q_{t}^{\T}N_{t}^{(k-1)}v_{t}'$,
its pathwise estimates immediately follow from Lemma~\ref{lem:LS-comp-tool},
using the good-event estimates $\norm{q_{t}},\norm{q_{t}'},\norm{v_{t}},\norm{v_{t}'}\lesssim d^{1/2}$,
the deterministic bound $\norm{N_{t}}\lesssim1$ from Lemma~\ref{lem:LS-comp-tool},
and $\norm{N_{t}'}\lesssim d^{1/2}$, $\norm{N_{t}''}\lesssim d$,
and $\norm{N_{t}'''}\lesssim d^{3/2}$. As for $H_{4}$, by Cauchy--Schwarz,
\[
\abs{H_{4}(t)}=\abs{q_{t}^{\T}N_{t}'''v_{t}}\leq\norm{q_{t}}\,\norm{N_{t}'''}\,\norm{v_{t}}\lesssim d^{5/2}\,.
\]

Finally, differentiating $E_{1}(t)=\inner{s_{t},W_{t}^{\alpha}s_{t}^{\circ2}}$
and using $s_{t}'=-s_{t}^{\circ2}$, 
\[
E_{1}'(t)=\alpha\inner{s_{t},(\ell_{t}'\circ W_{t}^{\alpha})\,s_{t}^{\circ2}}-\inner{s_{t}^{\circ2},W_{t}^{\alpha}s_{t}^{\circ2}}-2\inner{s_{t},W_{t}^{\alpha}s_{t}^{\circ3}}\,.
\]
Since $\norm{u_{t}}\lesssim d^{1/2}$ and $W_{t}^{1/2}\ell_{t}'=-N_{t}v_{t}$,
we can bound the first term as follows:
\begin{align*}
\sum_{i}w_{t,i}^{\alpha}\,\abs{\ell_{t,i}'}\,\abs{s_{t,i}}^{3} & \leq\bpar{\sum_{i}w_{t,i}^{\alpha}s_{t,i}^{2}}^{1/2}\bpar{\sum_{i}w_{t,i}^{\alpha}(\ell_{t,i}')^{2}s_{t,i}^{4}}^{1/2}=\norm{u_{t}}\bpar{\sum_{i}w_{t,i}(\ell_{t,i}')^{2}\rho_{t,i}^{2}s_{t,i}^{2}}^{1/2}\\
 & \lesssim d^{1/2}\,\bpar{\sum_{i}w_{t,i}(\ell_{t,i}')^{2}}^{1/2}=d^{1/2}\,\norm{N_{t}v_{t}}\lesssim d\,.
\end{align*}
The last two terms are bounded by $\norm{s_{t}}_{\infty}^{2}\,s_{t}^{\T}W_{t}^{\alpha}s_{t}\lesssim d$.
\end{proof}

\subsection{Concentration of base-point Gaussian polynomials}

At the base point, write $W:=W_{x}$, $w:=w_{x}$, $N:=N_{x}$, $\bar{\Lambda}:=\bar{\Lambda}_{x}$,
and note that 
\[
U=W_{x}^{\beta}A_{x}\,,\qquad U^{\T}U=I_{d}\,,\qquad P=UU^{\T}\,,\qquad B=W^{1/p}U=W_{x}^{1/2}A_{x}\,.
\]
For $u_{h}:=Uh$, we can write $v_{0}=Bh=W^{1/p}u_{h}$, $s_{0}=W^{-\beta}Uh$,
and $q_{0}=W^{-1/2}(Uh)^{\circ2}$. The three base-point terms in
\eqref{eq:hard-recursion} are Gaussian polynomials of degrees three,
four, and five:
\[
E_{1}(0)+\beta H_{1}(0)\,,\qquad H_{2}(0)=q_{0}^{\T}N_{0}'[h]v_{0}\,,\qquad H_{3}(0)=q_{0}^{\T}N_{0}''[h,h]v_{0}\,.
\]
The cubic term is handled by the original LS moment calculation and
the Gaussian Poincar\'e inequality. The quartic and quintic terms
will be handled via Wiener chaoses and some MSI calculations.

\subsubsection{The cubic Gaussian polynomial}
\begin{lem}
\label{lem:base-cubic} $\norm{E_{1}(0)+\beta H_{1}(0)}_{L^{2}}\lesssim d^{1/2}$,
so $E_{1}(0)+\beta H_{1}(0)\lesssim d^{1/2}$ with high probability.
\end{lem}

\begin{proof}
Define $M:=-W^{-\beta}U-\beta W^{-1/2}NW^{1/p}U=-W^{-1/2}(I+\beta N)W^{1/p}U$.
Then, 
\begin{align*}
E_{1}(0)+\beta H_{1}(0) & =\inner{s_{0},W_{0}^{\alpha}s_{0}^{\circ2}}+\beta q_{0}^{\T}N_{0}v_{0}=\inner{W^{-\beta}Uh,(Uh)^{\circ2}}+\beta\inner{W^{-1/2}(Uh)^{\circ2},NW^{1/p}Uh}\\
 & =\inner{(W^{-\beta}U+\beta W^{-1/2}NW^{1/p}U)\,h,(Uh)^{\circ2}}=-\inner{Mh,(Uh)^{\circ2}}:=-T(h)\,.
\end{align*}
Since $T$ is odd, $\E T=0$. By the Gaussian Poincar\'e inequality,
\[
\E[T^{2}]\leq\E[\norm{\nabla T}^{2}]=\E\bbrack{\norm{M^{\T}u_{h}^{\circ2}+2U^{\T}(u_{h}\circ Mh)}^{2}}\lesssim\E[\norm{M^{\T}u_{h}^{\circ2}}^{2}]+\E[\norm{U^{\T}(u_{h}\circ Mh)}^{2}]\,.
\]

As for the second term, denote the $i$-th row of $M$ by $m_{i}$.
Using $\E[(a^{\T}h)^{2}(b^{\T}h)^{2}]=\norm a^{2}\norm b^{2}+2\,(a^{\T}b)^{2}$
for deterministic vectors $a,b\in\Rd$ and $h\sim\msf N(0,I_{d})$
(see \cite[Proposition C.2]{KV24ipm}) with $\norm{e_{i}^{\T}U}^{2}=w_{i}$,
\[
\E[\norm{U^{\T}(u_{h}\circ Mh)}^{2}]\leq\E[\norm{u_{h}\circ Mh}^{2}]=\sum_{i}\E[u_{h,i}^{2}\,(m_{i}^{\T}h)^{2}]\leq3\sum_{i}w_{i}\,\norm{m_{i}}^{2}=3\tr(M^{\T}WM)\,.
\]
As for the first term, define $C_{ij}:=(MM^{T})_{ij}$ and write
\[
\norm{M^{\T}u_{h}^{\circ2}}^{2}=(u_{h}^{\circ2})^{\T}MM^{\T}u_{h}^{\circ2}=\sum_{ij}C_{ij}u_{h,i}^{2}u_{h,j}^{2}\,.
\]
Since $u_{h}=Uh$ is a centered Gaussian with covariance $\E[u_{h}u_{h}^{\T}]=P$,
$\E[u_{h,i}^{2}]=w_{i}$, and $\E[u_{h,i}u_{h,j}]=P_{ij}$. Thus,
\begin{align*}
\E[\norm{M^{\T}u_{h}^{\circ2}}^{2}] & =\sum_{ij}C_{ij}\E[u_{h,i}^{2}u_{h,j}^{2}]=\sum_{ij}C_{ij}w_{i}w_{j}+2\sum_{ij}C_{ij}P_{ij}^{2}=w^{\T}MM^{\T}w+2\tr(P^{\circ2}C)\\
 & =\norm{M^{\T}w}^{2}+2\tr(M^{\T}P^{\circ2}M)\leq\norm{M^{\T}w}^{2}+2\tr(M^{\T}WM)\,.
\end{align*}

We now bound $\tr(M^{\T}WM)$ and $\norm{M^{\T}w}^{2}$. As for the
trace, note that $W^{1/2}M=-(I+\beta N)W^{1/p}U$. Since $N\preceq pI_{d}$
and $\norm{W^{1/p}U}_{\frob}^{2}\leq\norm U_{\frob}^{2}=d$, we have
(ignoring any polylogarithmic factors) $\tr(M^{\T}WM)\lesssim d$.
As for $\norm{M^{\T}w}^{2}$, denoting $C:=(I+\beta N)\precsim I$,
\[
\norm{M^{\T}w}^{2}=w^{\T}W^{-1/2}CW^{1/p}UU^{\T}W^{1/p}CW^{-1/2}w\lesssim w^{\T}W^{-1}w=\sum_{i}w_{i}=d\,.
\]
Putting all these together, $\E[T^{2}]\lesssim d$, and the concentration
result in Lemma~\ref{lem:conc-gaussian-poly} proves the claim.
\end{proof}

\subsubsection{Multiple stochastic integral calculus\label{subsec:MSI}}

We now estimate the $L^{2}$-norm of $H_{2}(0)$ and $H_{3}(0)$,
for which we would need to bound the expectation of Gaussian polynomials
of degree $8$ and $10$ with numerous terms. This is practically
infeasible to carry out. We need a more principled and organized approach
to carry out this analysis, leveraging already well-developed theory.

To this end, we use the finite-dimensional Wiener chaos decomposition
of Gaussian polynomials. This is equivalently the Hermite expansion
under Gaussian measure. Precisely, one may want to represent a polynomial
with respect to some orthonormal polynomial basis with respect to
Gaussian measure, and then its squared $L^{2}$-norm is the sum of
the squared coefficients of the basis.

\paragraph{Wiener chaos calculus.}

One can interpret Gaussian polynomials as multiple stochastic integrals
(MSI), following Nualart \cite[\S1.1]{Nualart06Malliavin}. We refer
readers to \cite[\S3]{Nualart19malliavin} for a gentler introduction.
Suppose that $(B_{t})_{t\geq0}$ is a Brownian motion defined over
$(\Omega,\mc F,\P)$, with $\mc F$ generated by $(B_{t})_{t\geq0}$.
Let $L_{s}^{2}(\R_{+}^{n})$ be the space of symmetric $L^{2}$ functions
$f:\Rn_{+}\to\R$. For a non-symmetric $f:\R_{+}^{n}\to\R$, its symmetrization
is defined as $\tilde{f}(t_{1},\dots,t_{n})=\frac{1}{n!}\sum_{\sigma}f(t_{\sigma(1)},\dots,t_{\sigma(n)})$
where the sum runs over all permutations of $[n]$. 
\begin{defn}
[Multiple stochastic integral (MSI)]\label{def:MSI}For $f\in L_{s}^{2}(\R_{+}^{n})$,
its order-$n$ multiple stochastic integral is 
\[
I_{n}(f):=n!\,\int_{0}^{\infty}\int_{0}^{t_{n}}\cdots\int_{0}^{t_{2}}f(t_{1},\ldots,t_{n})\,\D B_{t_{1}}\cdots\D B_{t_{n}}\,.
\]
For non-symmetric $f$, we define $I_{n}(f):=I_{n}(\tilde{f})$.
\end{defn}

When $n=1$, this is the usual It\^o integral $I_{1}(f)=\int_{\R_{+}}f(t)\,\D B_{t}$.
As seen later, one can think of $I_{n}(f)$ as base Gaussian polynomials
of degree $n$. For $n,m\in\mathbb{N}$, the standard isometry of
multiple It\^o integrals says that for $f\in L^{2}(\R_{+}^{n})$
and $g\in L^{2}(\R_{+}^{m})$, 
\begin{equation}
\E[I_{n}(f)\,I_{m}(g)]=\ind_{n=m}\,n!\,\inner{\tilde{f},\tilde{g}}_{L^{2}(\R_{+}^{n})}\,.\label{eq:orthogonality}
\end{equation}
For $f\in L_{s}^{2}(\R_{+}^{n})$ and $g\in L_{s}^{2}(\R_{+}^{m})$,
their contraction $f\otimes_{r}g\in L^{2}(\R_{+}^{n+m-2r})$ of order
$r\in\{0,\ldots,n\wedge m\}$ is defined as
\[
(f\otimes_{r}g)(t_{1},\ldots,t_{n-r},s_{1},\ldots,s_{m-r})=\int_{\R_{+}^{r}}f(t_{1},\ldots,t_{n-r},x_{1},\ldots,x_{r})\,g(s_{1},\ldots,s_{m-r},x_{1},\ldots,x_{r})\,\D x_{1}\cdots\D x_{r}\,.
\]
The product formula is the basic rule that turns products of Gaussian
polynomials into orthogonal chaos pieces.
\begin{lem}[{Product formula, {\cite[Proposition~1.1.3]{Nualart06Malliavin}}}]
\label{lem:msi-product} For $f\in L_{s}^{2}(\R_{+}^{n})$ and $g\in L_{s}^{2}(\R_{+}^{m})$,
\[
I_{n}(f)\,I_{m}(g)=\sum_{r=0}^{n\wedge m}r!\,\binom{n}{r}\binom{m}{r}I_{n+m-2r}(f\otimes_{r}g)\,.
\]
\end{lem}

The next result, the Wiener chaos expansion, is the underlying principle
behind our attempt.
\begin{thm}[{Wiener chaos expansion, {\cite[Theorem~1.1.2]{Nualart06Malliavin}}}]
\label{thm:wiener-decom} Every $F\in L^{2}(\Omega)$ admits a unique
orthogonal expansion $F=\E F+\sum_{n=1}^{\infty}I_{n}(f_{n})$ for
suitable $f_{n}\in L_{s}^{2}(\R_{+}^{n})$. 
\end{thm}

Let $\mc H_{n}$ be the $L^{2}(\Omega)$-closure of the collection
of all order-$n$ MSI. The theorem says that $L^{2}(\Omega)=\oplus_{n=0}^{\infty}\mc H_{n}$.
Thus, after decomposing a Gaussian polynomial into its Wiener-chaos
pieces, its $L^{2}$ norm is the sum of the squared norms of those
pieces.

\paragraph{Tensor MSI.}

For the finite-dimensional Gaussian $h\sim\msf N(0,I_{d})$, choose
the orthonormal family $e_{i}(t)=\ind_{[i-1,i]}(t)$ in $L^{2}(\R_{+})$.
Then $h_{i}:=I_{1}(e_{i})=B_{i}-B_{i-1}\sim\msf N(0,1)$. A symmetric
$n$-tensor $T=(T_{i_{1},\ldots,i_{n}})\in(\Rd)^{\otimes n}$ can
be identified with the function
\begin{equation}
f_{T}:=\sum_{i_{1},\ldots,i_{n}=1}^{d}T_{i_{1},\ldots,i_{n}}e_{i_{1}}\otimes\cdots\otimes e_{i_{n}}\,,\label{eq:isometry}
\end{equation}
where $(e_{i_{1}}\otimes\cdots\otimes e_{i_{n}})(t_{1},\dots,t_{n}):=\prod_{j=1}^{n}e_{i_{j}}(t_{j})$,
and this isometry induces a tensor MSI as $I_{n}(T):=I_{n}(f_{T})$.
For a non-symmetric tensor $T$, we use the same notation for $I_{n}(\sym T)$.
We will use the following tensor forms of the MSI rules, deferring
their proof to \S\ref{app:tensor-msi-identities}.
\begin{lem}
\label{lem:tensor-I2} For a symmetric matrix $M$, $I_{2}(M)=h^{\T}Mh-\tr M$. 
\end{lem}

\begin{lem}
\label{lem:tensor-I3} For a symmetric third-order tensor $C\in(\Rd)^{\otimes3}$,
$I_{3}(C)=C[h^{\otimes3}]-3\inner{\tr C,h}$, where $\tr C\in\Rd$
is defined as $(\tr C)_{a}:=\sum_{b=1}^{d}C_{abb}$.
\end{lem}

\begin{lem}
\label{lem:tensor-I2-product} For symmetric matrices $M,N$, 
\[
I_{2}(M)\,I_{2}(N)=I_{4}(M\otimes N)+4I_{2}\bpar{\sym(MN)}+2\tr(MN)\,.
\]
\end{lem}

\begin{lem}
\label{lem:tensor-isometry} If $M,N$ are tensors of orders $n,m$,
then $\E[I_{n}(M)\,I_{m}(N)]=0$ for $n\neq m$, and 
\[
\norm{I_{n}(M)}_{L^{2}(\Omega)}^{2}=n!\,\norm{\sym M}_{\frob}^{2}\leq n!\,\norm M_{\frob}^{2}\,.
\]
\end{lem}

\subsubsection{The quartic Gaussian polynomial}

For the quartic term, we introduce the notation $N_{k}:=\Dd N(x)[e_{k}]$
for $k\in[d]$. Also, for each $i\in[m]$, define matrices 
\[
(M_{i})_{k\ell}:=e_{i}^{\T}N_{k}Be_{\ell}\,,\qquad S_{i}:=\sym M_{i}\,.
\]
Then,
\[
e_{i}^{\T}\Dd N[h]\,Bh=\sum_{k\ell}h_{k}h_{\ell}e_{i}^{\T}N_{k}Be_{\ell}=h^{\T}M_{i}h=h^{\T}S_{i}h\,.
\]

\paragraph{Decomposition of $H_{2}(0)$.}

We now decompose $H_{2}(0)$ into its Wiener-chaos pieces. Define
$K:=W^{-1/2}U=W^{-1/p}A_{x}$ (so $KK^{\T}=W_{x}^{-1/2}U_{x}U_{x}^{\T}W_{x}^{-1/2}$
and $(KK^{\T})_{ii}=1$), and denote each row by $k_{i}$ (so $\norm{k_{i}}=1$).
Let $K_{i}:=k_{i}k_{i}^{\T}$. By Lemma~\ref{lem:tensor-I2} and
$\norm{k_{i}}=1$,
\[
I_{2}(K_{i})=h^{\T}K_{i}h-\tr K_{i}=h^{\T}K_{i}h-1\,.
\]

Recall that $H_{2}(0)=q^{\T}N'[h]v$ for $q=W_{x}^{1/2}(W_{x}^{-1/p}A_{x}h)^{\circ2}$
and $v=W_{x}^{1/2}A_{x}h=B_{x}h$. Hence, $q_{i}=w_{i}^{1/2}(k_{i}^{\T}h)^{2}=w_{i}^{1/2}\,(1+I_{2}(K_{i}))$
and $e_{i}^{\T}N'[h]Bh=h^{\T}S_{i}h=\tr S_{i}+I_{2}(S_{i})$. Using
Lemma~\ref{lem:tensor-I2-product},
\begin{align*}
H_{2}(0) & =\sum_{i}w_{i}^{1/2}\,\bpar{1+I_{2}(K_{i})}\bpar{\tr S_{i}+I_{2}(S_{i})}\\
 & =\sum_{i}w_{i}^{1/2}\,\bpar{I_{2}(K_{i})\,I_{2}(S_{i})+\tr S_{i}\,I_{2}(K_{i})+I_{2}(S_{i})+\tr S_{i}}\\
 & =\sum_{i}w_{i}^{1/2}\,\Bpar{I_{4}(K_{i}\otimes S_{i})+4I_{2}\bpar{\sym(K_{i}S_{i})}+2\tr(K_{i}S_{i})+\tr S_{i}\,I_{2}(K_{i})+I_{2}(S_{i})+\tr S_{i}}\\
 & =:F_{4}+F_{2}+F_{0}\,,
\end{align*}
where
\begin{align*}
F_{4} & =I_{4}(T_{4})\quad\text{for }T_{4}:=\sum_{i}w_{i}^{1/2}K_{i}\otimes S_{i}\,,\\
F_{2} & =I_{2}(C_{2})\quad\text{for }C_{2}:=\sum_{i}w_{i}^{1/2}\,\bpar{4\sym(K_{i}S_{i})+\tr(S_{i})\,K_{i}+S_{i}}\,,\\
F_{0} & =\sum_{i}w_{i}^{1/2}\,\bpar{2\tr(K_{i}S_{i})+\tr S_{i}}\,.
\end{align*}
Since $\norm{H_{2}(0)}^{2}=\norm{F_{4}}^{2}+\norm{F_{2}}^{2}+\norm{F_{0}}^{2}$
due to the orthogonality of $\{I_{n}\}$ \eqref{eq:orthogonality},
it suffices to estimate the $L^{2}$-norm of $F_{4},F_{2},F_{0}$.
By Lemma~\ref{lem:tensor-isometry}, $\norm{F_{4}}^{2}=\norm{I_{4}(T_{4})}^{2}\leq24\,\norm{T_{4}}_{\frob}^{2}$
and $\norm{F_{2}}^{2}=\norm{I_{2}(C_{2})}^{2}=2\,\norm{C_{2}}_{\frob}^{2}$,
so the remaining task is to bound the Frobenius norms of structured
tensors such as $T_{4}$ and $C_{2}$.

\paragraph{Useful estimates.}

We extract essential terms that repeatedly appear in the tensor analysis:
\[
\mc E_{1}:=\sum_{k=1}^{d}\norm{N_{k}B}_{\frob}^{2}\,,\qquad\mc T_{1}:=\sum_{i}(\tr S_{i})^{2}\,.
\]
We will show $\mc E_{1}\lesssim d$ and $\mc T_{1}\lesssim d^{2}$. 
\begin{lem}
\label{lem:Nr-cancellation} It holds that
\[
\bar{\Lambda}w^{1/2}=Nw^{1/2}=0\,,\qquad N_{a}w^{1/2}=\half\,N^{2}Be_{a}\,,\qquad\sum_{i}w_{i}^{1/2}S_{i}=\half\,B^{\T}N^{2}B\,.
\]
\end{lem}

\begin{proof}
Since $P^{\circ2}\ind=w$,
\[
\bar{\Lambda}w^{1/2}=(I-W^{-1/2}P^{\circ2}W^{-1/2})\,w^{1/2}=w^{1/2}-W^{-1/2}P^{\circ2}\ind=0\,.
\]
Since $(I-c_{p}\bar{\Lambda})\,w^{1/2}=w^{1/2}$, we also have $(I-c_{p}\bar{\Lambda})^{-1}w^{1/2}=w^{1/2}$.
Then, $Nw^{1/2}=2\bar{\Lambda}(1-c_{p}\bar{\Lambda})^{-1}w^{1/2}=2\bar{\Lambda}w^{1/2}=0$.
Moreover, using $\Dd w[h]=-W^{1/2}NW^{\half}A_{x}h$ (Lemma~\ref{lem:DWh}),
\begin{equation}
\omega_{a}:=\Dd w^{1/2}[e_{a}]=\half\,W^{-1/2}\Dd w[e_{a}]=-\half\,NBe_{a}\,.\label{eq:omega-a}
\end{equation}
Differentiating $Nw^{1/2}=0$ in direction $e_{a}$ gives
\[
N_{a}w^{1/2}=-N\omega_{a}=\half\,N^{2}Be_{a}\,.
\]
Using this identity,
\begin{align*}
\Bpar{\sum_{i}w_{i}^{1/2}S_{i}}_{k\ell} & =\half\sum_{i}w_{i}^{1/2}\,(e_{i}^{\T}N_{k}Be_{\ell}+e_{i}^{\T}N_{\ell}Be_{k})=\half\,\bpar{(w^{1/2})^{\T}N_{k}Be_{\ell}+(w^{1/2})^{\T}N_{\ell}Be_{k}}\\
 & =\half\,\bpar{\half\,e_{k}^{\T}B^{\T}N^{2}Be_{\ell}+\half\,e_{\ell}^{\T}B^{\T}N^{2}Be_{k}}=\half\,e_{k}^{\T}B^{\T}N^{2}Be_{\ell}=\half\,(B^{\T}N^{2}B)_{k\ell}\,.
\end{align*}
This completes the proof.
\end{proof}
Before proceeding, note that $\norm B\leq1$ and $\norm B_{\frob}^{2}\leq d$
from $B^{\T}B=U^{\T}W^{2/p}U\preceq U^{\T}U=I_{d}$.
\begin{lem}
\label{lem:mcE1} $\mc E_{1}=\sum_{k=1}^{d}\norm{N_{k}B}_{\frob}^{2}\lesssim d$.
\end{lem}

\begin{proof}
Fix a direction $r\in\Rd$ and consider a path $x_{t}=x+tr$. Using
the computations along the path as in Lemma~\ref{lem:good-event},
for $z[r]=\beta\Dd\ell_{x}[r]-A_{x}r$,
\[
\norm{\Dd N[r]}_{\frob}\lesssim\norm{\Dd\Lambda[r]}_{\frob}\lesssim\norm{\Dd\Phi[r]}_{\frob}\lesssim\norm{\Dd U[r]}_{\frob}\leq\norm{z[r]}_{W}\,.
\]
Since $\norm B\leq1$, $\norm N\lesssim1$, and $W\preceq I$,
\begin{align*}
\norm{\Dd N[r]\,B}_{\frob} & \leq\norm{\Dd N[r]}_{\frob}\norm B\lesssim\norm{W^{1/2}z[r]}=\norm{(I+\beta N)\,W^{1/2}A_{x}r}\\
 & \lesssim\norm{W^{1/2}A_{x}r}\leq\sqrt{r^{\T}A_{x}^{\T}W_{x}^{\alpha}A_{x}r}=\norm r\,.
\end{align*}
Substituting $r=e_{k}$, we obtain $\norm{N_{k}B}_{\frob}\lesssim1$,
so squaring and summing over $k$ gives $\sum_{k}\norm{N_{k}B}_{\frob}^{2}\lesssim d$.
\end{proof}
\begin{lem}
\label{lem:mcT1} $\mc T_{1}=\sum_{i}(\tr S_{i})^{2}\lesssim d^{2}$.
\end{lem}

\begin{proof}
For $B_{k}:=Be_{k}$ and $v:=\sum_{k}N_{k}B_{k}\in\mathbb{R}^{m}$,
then the $i$-th coordinate of $v$ is exactly $v_{i}=\tr S_{i}$.
Hence,
\[
\mc T_{1}=\sum_{i}(\tr S_{i})^{2}=\sum_{i}v_{i}^{2}=\Bnorm{\sum_{k}N_{k}B_{k}}^{2}\leq d\sum_{k}\norm{N_{k}B_{k}}^{2}\leq d\sum_{k}\norm{N_{k}B}_{\frob}^{2}\lesssim d^{2}\,,
\]
where the last step is Lemma~\ref{lem:mcE1}.
\end{proof}
\begin{lem}
\label{lem:S-contractions}$\sum_{i}\norm{S_{i}}_{\frob}^{2}\lesssim d$
and $\sum_{i}(k_{i}^{\T}S_{i}k_{i})^{2}\lesssim d$.
\end{lem}

\begin{proof}
For the Frobenius bound, by the definition of $S_{i}$ as the symmetrization
of $M_{i}$,
\[
\sum_{i}\norm{S_{i}}_{\frob}^{2}\leq\sum_{i}\norm{M_{i}}_{\frob}^{2}=\sum_{i,k,\ell}(e_{i}^{\T}N_{k}Be_{\ell})^{2}=\sum_{k}\norm{N_{k}B}_{\frob}^{2}\lesssim d
\]
by Lemma~\ref{lem:mcE1}. Next, using $\norm{k_{i}}=1$,
\begin{align*}
k_{i}^{\T}S_{i}k_{i} & =k_{i}^{\T}M_{i}k_{i}=\sum_{j,\ell}e_{i}^{\T}N_{j}Be_{\ell}\,k_{ij}k_{i\ell}=e_{i}^{\T}N'[k_{i}]Bk_{i}\\
 & =\sum_{k}k_{ik}\,e_{i}^{\T}N_{k}Bk_{i}\leq\Bpar{\sum_{k}(e_{i}^{\T}N_{k}Bk_{i})^{2}}^{1/2}\leq\Bpar{\sum_{k}\norm{e_{i}^{\T}N_{k}B}^{2}}^{1/2}\,.
\end{align*}
Squaring and summing over $i$ gives
\[
\sum_{i}\bpar{k_{i}^{\T}S_{i}k_{i}}^{2}\leq\sum_{i,k}\norm{e_{i}^{\T}N_{k}B}^{2}=\sum_{k}\norm{N_{k}B}_{\frob}^{2}\lesssim d\,.
\]
This completes the proof.
\end{proof}

\paragraph{Norm of $H_{2}(0)$.}

Combining the computational lemmas above, we can readily bound $\norm{H_{2}(0)}_{L^{2}}$.
\begin{lem}
\label{lem:qNv-l2} $\norm{H_{2}(0)}_{L^{2}}\lesssim d$, so $H_{2}(0)=\Otilde(d)$
with high probability.
\end{lem}

\begin{proof}
We estimate the three orthogonal chaos components $F_{4},F_{2},F_{0}$.
By Lemma~\ref{lem:tensor-isometry},
\[
\norm{F_{4}}_{L^{2}}^{2}\leq24\norm{T_{4}}_{\frob}^{2}\,,\qquad\norm{F_{2}}_{L^{2}}^{2}=2\norm{C_{2}}_{\frob}^{2}\,.
\]

For $F_{4}$,
\[
\norm{T_{4}}_{\frob}^{2}=\Bnorm{\sum_{i}w_{i}^{1/2}K_{i}\otimes S_{i}}_{\frob}^{2}=\sum_{i,j}\sqrt{w_{i}w_{j}}\,\inner{K_{i},K_{j}}\inner{S_{i},S_{j}}\,.
\]
Observe that
\begin{equation}
\sqrt{w_{i}w_{j}}\inner{K_{i},K_{j}}=\sqrt{w_{i}w_{j}}\,(k_{i}^{\T}k_{j})^{2}=\frac{P_{ij}^{2}}{\sqrt{w_{i}w_{j}}}=\bpar{W^{-1/2}P^{\circ2}W^{-1/2}}_{ij}\,.\label{eq:wwKK}
\end{equation}
Set $G:=W^{-1/2}P^{\circ2}W^{-1/2}$. The preceding identity gives
$\sqrt{w_{i}w_{j}}\inner{K_{i},K_{j}}=G_{ij}$. Here we use the kernel
inequality in the following form. If $H_{ij}:=\inner{A_{i},A_{j}}$
is any positive-semidefinite Gram matrix, then $0\preceq G\preceq I$
implies
\begin{equation}
\sum_{i,j}G_{ij}\inner{A_{i},A_{j}}=\tr(GH)\leq\tr H=\sum_{i}\norm{A_{i}}^{2}\,.\label{eq:gaa-bound}
\end{equation}
Applying this with $A_{i}=S_{i}$ and using the first estimate in
Lemma~\ref{lem:S-contractions} 
\[
\norm{T_{4}}_{\frob}^{2}\leq\sum_{i}\norm{S_{i}}_{\frob}^{2}\lesssim d\,.
\]
Thus, $\norm{F_{4}}_{L^{2}}^{2}\lesssim d$.

For $F_{2}$, recall $C_{2}=\sum_{i}w_{i}^{1/2}\,(S_{i}+4\sym(K_{i}S_{i})+\tr(S_{i})K_{i})$.
Note that Lemma~\ref{lem:Nr-cancellation} yields
\begin{equation}
\sum_{i}w_{i}^{1/2}\tr S_{i}=\tr\Bpar{\sum_{i}w_{i}^{1/2}S_{i}}=\half\tr(B^{\T}N^{2}B)\lesssim\tr(B^{\T}B)\lesssim d\,.\label{eq:wtrS}
\end{equation}
For the first summand, we again use Lemma~\ref{lem:Nr-cancellation}:
\[
\Bnorm{\sum_{i}w_{i}^{1/2}S_{i}}_{\frob}=\half\,\norm{B^{\T}N^{2}B}_{\frob}\lesssim\tr(B^{\T}B)\lesssim d\,.
\]
For the second summand, by $\sum_{i}w_{i}\lesssim d$ and the first
estimate of Lemma~\ref{lem:S-contractions},
\[
\Bnorm{\sum_{i}w_{i}^{1/2}\sym(K_{i}S_{i})}_{\frob}\leq\sum_{i}w_{i}^{1/2}\norm{K_{i}S_{i}}_{\frob}\leq\sum_{i}w_{i}^{1/2}\norm{S_{i}}_{\frob}\leq\Bpar{\sum_{i}w_{i}}^{1/2}\Bpar{\sum_{i}\norm{S_{i}}_{\frob}^{2}}^{1/2}\lesssim d\,.
\]
For the trace summand, by \eqref{eq:wwKK} and Lemma~\ref{lem:mcT1},
\[
\Bnorm{\sum_{i}w_{i}^{1/2}K_{i}\tr S_{i}}_{\frob}^{2}=\sum_{i,j}w_{i}^{1/2}w_{j}^{1/2}\inner{K_{i},K_{j}}\tr S_{i}\,\tr S_{j}\leq\sum_{i}(\tr S_{i})^{2}\lesssim d^{2}\,.
\]
 Hence, $\norm{C_{2}}_{\frob}\lesssim d$, and $\norm{F_{2}}_{L^{2}}^{2}\lesssim d^{2}$.

For $F_{0}$, recall $F_{0}=\sum_{i}w_{i}^{1/2}\,(2\tr(K_{i}S_{i})+\tr S_{i})$.
Since $\tr(K_{i}S_{i})=k_{i}^{\T}S_{i}k_{i}$, Lemma~\ref{lem:S-contractions}
leads to

\[
\sum_{i}w_{i}^{1/2}\tr(K_{i}S_{i})\leq\Bpar{\sum_{i}w_{i}}^{1/2}\Bpar{\sum_{i}(k_{i}^{\T}S_{i}k_{i})^{2}}^{1/2}\lesssim d\,.
\]
The second term is bounded by $O(d)$ as in \eqref{eq:wtrS}. Hence,
$\norm{F_{0}}_{L^{2}}^{2}=|F_{0}|^{2}\lesssim d^{2}$.

Combining the orthogonal chaos pieces,
\[
\norm{H_{2}(0)}_{L^{2}}^{2}=\norm{F_{4}}_{L^{2}}^{2}+\norm{F_{2}}_{L^{2}}^{2}+\norm{F_{0}}_{L^{2}}^{2}\lesssim d^{2}\,.
\]
The high-probability statement follows from Lemma~\ref{lem:conc-gaussian-poly}.
\end{proof}

\subsubsection{The quintic Gaussian polynomial\label{sec:quintic-gaussian-polynomial}}

Define $N_{ab}:=\Dd^{2}N_{x}[e_{a},e_{b}]$ for $a,b\in[d]$. Also,
for each row $i$, define the order-$3$ tensor $C_{i}^{0}$ by 
\[
C_{i}^{0}(a,b,\ell):=e_{i}^{\T}N_{ab}Be_{\ell}\,,\qquad C_{i}:=\sym C_{i}^{0}\,.
\]
With this notation, $C_{i}[h^{\otimes3}]=e_{i}^{\T}N_{0}''[h,h]Bh$.
We also define $\tau_{i}:=\tr C_{i}$ and $(\tau_{i})_{a}:=\sum_{b}C_{i}[a,b,b]$.

\paragraph{Decomposition of $H_{3}(0)$.}

As in the quartic term, note that $(k_{i}^{\T}h)^{2}=h^{\T}K_{i}h=1+I_{2}(K_{i})$.
Also, $C_{i}[h^{\otimes3}]=I_{3}(C_{i})+3I_{1}(\tau_{i})$ by Lemma~\ref{lem:tensor-I3}.
Thus,
\[
H_{3}(0)=\sum_{i}w_{i}^{1/2}\,\bpar{1+I_{2}(K_{i})}\bpar{I_{3}(C_{i})+3I_{1}(\tau_{i})}\,.
\]
By Lemma~\ref{lem:msi-product},
\[
I_{2}(K_{i})I_{3}(C_{i})=I_{5}(K_{i}\otimes C_{i})+6I_{3}(K_{i}\otimes_{1}C_{i})+6I_{1}(K_{i}\otimes_{2}C_{i})\,,\quad I_{2}(K_{i})I_{1}(\tau_{i})=I_{3}(K_{i}\otimes\tau_{i})+2I_{1}(K_{i}\otimes_{1}\tau_{i})\,.
\]
This implies that 
\begin{align*}
\bpar{1+I_{2}(K_{i})}\bpar{I_{3}(C_{i})+3I_{1}(\tau_{i})} & =I_{5}(K_{i}\otimes C_{i})+I_{3}(C_{i}+6K_{i}\otimes_{1}C_{i}+3K_{i}\otimes\tau_{i})+\\
 & \quad+I_{1}(3\tau_{i}+6K_{i}\otimes_{2}C_{i}+6K_{i}\otimes_{1}\tau_{i})\,.
\end{align*}
Hence, we can write $H_{3}(0)=I_{5}(T_{5})+I_{3}(T_{3})+I_{1}(T_{1})$,
where
\begin{align*}
T_{5} & :=\sum_{i}w_{i}^{1/2}K_{i}\otimes C_{i}\,,\\
T_{3} & :=\sum_{i}w_{i}^{1/2}\,(C_{i}+6K_{i}\otimes_{1}C_{i}+3K_{i}\otimes\tau_{i})\,,\\
T_{1} & :=\sum_{i}w_{i}^{1/2}\,(3\tau_{i}+6K_{i}\otimes_{2}C_{i}+6K_{i}\otimes_{1}\tau_{i})\,.
\end{align*}

\paragraph{Useful estimates.}

As in the analysis of $H_{2}(0)$, we first focus on bounding useful
terms 
\[
\mc E_{2}:=\sum_{a,b=1}^{d}\norm{N_{ab}B}_{\frob}^{2}\,,\qquad\mc T_{2}:=\sum_{i}\norm{\tr C_{i}}^{2}\,,\qquad S_{\tau}:=\sum_{i}w_{i}^{1/2}\tr C_{i}\,.
\]
To this end, we need two computational lemmas.
\begin{lem}
\label{lem:two-direction-score-basis}$\mc Z_{2}:=\sum_{a,b}\norm{W^{1/2}z_{ab}}^{2}\lesssim d^{2}$
and $\mc B_{2}:=\sum_{a,b}\norm{B_{a,b}}^{2}\lesssim d^{2}$.
\end{lem}

\begin{proof}
Let $C_{0}:=I+\beta N$, $B_{a}:=Be_{a}$, and $z_{a}:=z[e_{a}]$.
Recall from \eqref{eq:qt-ztWt} that $W^{1/2}z_{a}=-C_{0}W^{1/2}A_{x}e_{a}=-C_{0}B_{a}$.
Differentiating this identity in direction $-e_{b}$,
\[
-\frac{1}{2}\,\ell_{b}\circ W^{1/2}z_{a}-W^{1/2}z_{ab}=\beta N_{b}B_{a}+C_{0}B_{a,b}\,,
\]
where $B_{a,b}=\Dd B[e_{b}]e_{a}=\Diag(\half\,\ell_{b}-s_{b})\,B_{a}$.
Substituting this and $W^{1/2}z_{a}=-C_{0}B_{a}$ yields
\[
-W^{1/2}z_{ab}=\beta N_{b}B_{a}+C_{0}\,\bpar{(\half\,\ell_{b}-s_{b})\circ B_{a}}-\half\,\ell_{b}\circ C_{0}B_{a}\,.
\]
Note that $\norm{e_{b}}_{A^{\T}WA}\leq\norm{e_{b}}_{g_{0}}=1$. By
Lemma~\ref{lem:weight-close} with \eqref{eq:sigma-rho-derv}, $\norm{\ell_{b}}_{\infty},\norm{s_{b}}_{\infty}\lesssim1$,
and recall $\norm B_{\frob}^{2}\lesssim d$. Thus,
\[
\mc B_{2}=\sum_{a,b}\norm{B_{a,b}}^{2}=\sum_{b}\bnorm{\Diag(\half\,\ell_{b}-s_{b})\,B}_{\frob}^{2}\lesssim d^{2}\,.
\]
Since $\norm{N_{b}}\lesssim1$ (Lemma~\ref{lem:LS-comp-tool}) and
$\mc E_{1}=\sum_{k}\norm{N_{k}B}_{\frob}^{2}\lesssim d$ (Lemma~\ref{lem:mcE1}),
\[
\mc Z_{2}=\sum_{a,b}\norm{W^{1/2}z_{ab}}^{2}\lesssim\sum_{b}\norm{N_{b}B}_{\frob}^{2}+\mc B_{2}+\sum_{b}\norm{\ell_{b}}_{\infty}^{2}\norm{C_{0}B}_{\frob}^{2}\lesssim d^{2}\,.
\]
This completes the proof.
\end{proof}
\begin{lem}
\label{lem:moving-frame-second-derivatives}It holds that
\[
\sum_{a}\norm{U_{a}}_{\frob}^{2}\lesssim d\,,\qquad\sum_{a}\norm{e_{i}^{\T}U_{a}}^{2}\lesssim d\,w_{i}\,,\qquad\sum_{i}\sum_{a}\norm{e_{i}^{\T}U_{a}}^{2}\lesssim d\,,\qquad\sum_{a,b}\norm{U_{ab}}_{\frob}^{2}\lesssim d^{2}\,.
\]
\end{lem}

\begin{proof}
By \eqref{eq:Ut-derivatives} and \eqref{eq:ptperp-zt-ut}, $U_{a}=P^{\perp}Z_{a}U$
and $\norm{U_{a}}_{\frob}^{2}\leq\norm{z_{a}}_{W}^{2}$. Using $\norm{s_{a}}_{\infty},\norm{\ell_{a}}_{\infty}\lesssim1$
and $\norm{C_{0}}\lesssim1$, 
\[
\sum_{a}\norm{U_{a}}_{\frob}^{2}\leq\sum_{a}\norm{z_{a}}_{W}^{2}=\sum_{a}\norm{C_{0}B_{a}}^{2}\lesssim\norm B_{\frob}^{2}\lesssim d\,.
\]
This implies that $\sum_{i}\sum_{a}\norm{e_{i}^{\T}U_{a}}^{2}=\sum_{a}\norm{U_{a}}_{\frob}^{2}\lesssim d$.
Next, as $\norm{e_{i}^{\T}U_{a}}\lesssim\norm{z_{a}}_{\infty}\,w_{i}^{1/2}\lesssim w_{i}^{1/2}$
from \eqref{eq:u'ti-bound}, we obtain that $\sum_{a}\norm{e_{i}^{\T}U_{a}}^{2}\lesssim dw_{i}$.

As for the last item, define $U_{ab}^{\perp}:=P^{\perp}U_{ab}$. By
Lemma~\ref{lem:normal-frame-calculus},
\[
U_{ab}^{\perp}=P^{\perp}(P^{\perp}Z_{a}U)_{b}=-P^{\perp}P_{b}Z_{a}U+P^{\perp}Z_{ab}U+P^{\perp}Z_{a}U_{b}\,.
\]
Since $\norm{P_{b}}_{\frob}\lesssim\norm{U_{b}}_{\frob}$ from $P_{b}=U_{b}U^{\T}+UU_{b}^{\T}$,
by Lemma~\ref{lem:two-direction-score-basis}, 
\begin{align*}
\sum_{a,b}\norm{P^{\perp}P_{b}Z_{a}U}_{\frob}^{2} & \lesssim d\sum_{b}\norm{U_{b}}_{\frob}^{2}\lesssim d^{2}\,,\\
\sum_{a,b}\norm{P^{\perp}Z_{ab}U}_{\frob}^{2} & \leq\sum_{a,b}\norm{Z_{ab}U}_{\frob}^{2}=\sum_{a,b}\norm{W^{1/2}z_{ab}}^{2}=\mc Z_{2}\lesssim d^{2}\,,\\
\sum_{a,b}\norm{P^{\perp}Z_{a}U_{b}}_{\frob}^{2} & \lesssim d\sum_{b}\norm{U_{b}}_{\frob}^{2}\lesssim d^{2}\,.
\end{align*}
Thus, $\sum_{a,b}\norm{P^{\perp}U_{ab}}_{\frob}^{2}\lesssim d^{2}$.
By Lemma~\ref{lem:normal-frame-calculus} again, we can take a suitable
$U$ so that
\[
\sum_{a,b}\norm{U_{ab}}_{\frob}^{2}\lesssim\sum_{a,b}\norm{P^{\perp}U_{ab}}_{\frob}^{2}+\Bpar{\sum_{a}\norm{U_{a}}_{\frob}^{2}}^{2}\lesssim d^{2}\,.
\]
This proves the lemma.
\end{proof}
We now bound the useful quantities.
\begin{lem}
\label{lem:two-direction-derivatives}$\mc E_{2}=\sum_{a,b}\norm{N_{ab}B}_{\frob}^{2}\lesssim d^{2}$.
\end{lem}

\begin{proof}
For $\phi_{i}=R(u_{i})=u_{i}\otimes u_{i}/\norm{u_{i}}$, we can check
that $\Dd^{2}\phi_{i}[e_{a},e_{b}]=\Dd R(u_{i})[u_{i,ab}]+\Dd^{2}R(u_{i})[u_{i,a},u_{i,b}]$.
By Lemma~\ref{lem:row-map}, 
\[
\sum_{a,b}\norm{\Dd^{2}\phi_{i}[e_{a},e_{b}]}^{2}\lesssim\sum_{a,b}\norm{u_{i,ab}}^{2}+\frac{1}{w_{i}}\,\Bpar{\sum_{a}\norm{u_{i,a}}^{2}}^{2}=\sum_{a,b}\norm{e_{i}^{\T}U_{ab}}_{\frob}^{2}+\frac{1}{w_{i}}\,\Bpar{\sum_{a}\norm{e_{i}^{\T}U_{a}}^{2}}^{2}\,.
\]
By Lemma~\ref{lem:moving-frame-second-derivatives},
\[
\sum_{a,b}\norm{\Phi_{ab}}_{\frob}^{2}=\sum_{i}\sum_{a,b}\norm{\Dd^{2}\phi_{i}[e_{a},e_{b}]}^{2}\lesssim\sum_{a,b}\norm{U_{ab}}_{\frob}^{2}+\sum_{i}d\sum_{a}\norm{e_{i}^{\T}U_{a}}^{2}\lesssim d^{2}\,.
\]

Differentiating $\bar{\Lambda}=I-\Phi\Phi^{\T}$ gives $\bar{\Lambda}_{ab}=-\Phi_{ab}\Phi^{\T}-\Phi_{b}\Phi_{a}^{\T}-\Phi_{a}\Phi_{b}^{\T}-\Phi\Phi_{ab}^{\T}$.
Since $\norm{\Phi}\leq1$, the summation of the Frobenius norm of
the second-derivative terms is bounded by $O(d^{2})$. Due to \eqref{eq:Phi-derivative-bounds},
$\norm{\Phi_{a}}\lesssim\norm{U_{a}}$. Hence,
\[
\sum_{a,b}\norm{\Phi_{b}\Phi_{a}^{\T}}_{\frob}^{2}\leq\Bpar{\sum_{a}\norm{\Phi_{a}}_{\frob}^{2}}^{2}\lesssim\Bpar{\sum_{a}\norm{U_{a}}_{\frob}^{2}}^{2}\lesssim d^{2}\,.
\]
Hence, $\sum_{a,b}\norm{\bar{\Lambda}_{ab}}_{\frob}^{2}\lesssim d^{2}$.

Next, differentiating $N=2\bar{\Lambda}R$ in $e_{a}$ and $e_{b}$
yields $N_{ab}=2R\bar{\Lambda}_{ab}R+2c_{p}R\bar{\Lambda}_{a}R\bar{\Lambda}_{b}R+2c_{p}R\bar{\Lambda}_{b}R\bar{\Lambda}_{a}R$.
Then,
\[
\sum_{a,b}\norm{R\bar{\Lambda}_{ab}RB}_{\frob}^{2}\leq\norm R^{2}\,\norm{RB}^{2}\sum_{a,b}\norm{\bar{\Lambda}_{ab}}_{\frob}^{2}\lesssim\sum_{a,b}\norm{\bar{\Lambda}_{ab}}_{\frob}^{2}\lesssim d^{2}\,.
\]
For the quadratic terms, using $N_{a}=2R\bar{\Lambda}_{a}R$ and $\norm{R^{-1}}\leq1$,
we have $R\bar{\Lambda}_{a}R\bar{\Lambda}_{b}RB=\frac{1}{4}\,N_{a}R^{-1}N_{b}B$.
Since $\norm{N_{a}}\lesssim1$ (Lemma~\ref{lem:LS-comp-tool}) and
$\mc E_{1}\lesssim d$,
\[
\sum_{a,b}\norm{N_{a}R^{-1}N_{b}B}_{\frob}^{2}\leq\sum_{a}\norm{N_{a}R^{-1}}^{2}\sum_{b}\norm{N_{b}B}_{\frob}^{2}\lesssim d^{2}\,.
\]
Therefore, $\mc E_{2}=\sum_{a,b}\norm{N_{ab}B}_{\frob}^{2}\lesssim d^{2}$.
\end{proof}
\begin{lem}
\label{lem:T2-bound}$\mc T_{2}=\sum_{i}\norm{\tr C_{i}}^{2}\lesssim d^{3}$.
\end{lem}

\begin{proof}
Using $N_{ab}=N_{ba}$, we have $C_{i}(a,b,c)=\frac{1}{3}\,(e_{i}^{\T}N_{ab}B_{c}+e_{i}^{\T}N_{bc}B_{a}+e_{i}^{\T}N_{ca}B_{b})$.
Thus, for $\tau_{i}=\tr C_{i}$,
\[
\tau_{i}(a)=\sum_{b}C_{i}(a,b,b)=\frac{2}{3}\,e_{i}^{\T}\sum_{b}N_{ab}B_{b}+\frac{1}{3}\,e_{i}^{\T}\sum_{b}N_{bb}B_{a}=:e_{i}^{\T}(\frac{2}{3}\,\Gamma_{a}+\frac{1}{3}\,\Delta_{a})\,.
\]
Hence, $\mc T_{2}=\sum_{i}\norm{\tau_{i}}^{2}\lesssim\sum_{a}\norm{\Gamma_{a}}^{2}+\sum_{a}\norm{\Delta_{a}}^{2}$,
where
\begin{align*}
\sum_{a}\norm{\Gamma_{a}}^{2} & \leq d\sum_{a,b}\norm{N_{ab}B_{b}}^{2}\leq d\sum_{a,b}\norm{N_{ab}B}_{\frob}^{2}\leq d\mc E_{2}\\
\sum_{a}\norm{\Delta_{a}}^{2} & =\Bnorm{\sum_{b}N_{bb}B}_{\frob}^{2}\leq d\sum_{b}\norm{N_{bb}B}_{\frob}^{2}\leq d\mc E_{2}\,.
\end{align*}
Therefore, $\mc T_{2}=\sum_{i}\norm{\tau_{i}}^{2}\lesssim d^{3}$.
\end{proof}
\begin{lem}
\label{lem:S-tau-bound}$\norm{S_{\tau}}^{2}\lesssim d^{3}$ for $S_{\tau}=\sum_{i}w_{i}^{1/2}\tr C_{i}$.
\end{lem}

\begin{proof}
Recall from the previous lemma that with $\Gamma_{a}=\sum_{b}N_{ab}B_{b}$,
$\Delta_{a}=\sum_{b}N_{bb}B_{a}$, and $\omega:=w^{1/2}$,
\[
(S_{\tau})_{a}=\omega^{\T}\bpar{\frac{2}{3}\Gamma_{a}+\frac{1}{3}\Delta_{a}}=\frac{2}{3}\sum_{b}B_{b}^{\T}N_{ab}\omega+\frac{1}{3}\sum_{b}B_{a}^{\T}N_{bb}\omega\,.
\]
Using $\bar{\Lambda}\omega=0$, $N\omega=0$, and $N_{a}\omega=-N\omega_{a}=\half N^{2}B_{a}$
(Lemma~\ref{lem:Nr-cancellation}), and $\omega_{b}=-\half NB_{b}$
\eqref{eq:omega-a},
\[
N_{ab}\omega=\half\,(N_{a}NB_{b}+N_{b}NB_{a}+NN_{b}B_{a}+N^{2}B_{a,b})\,.
\]
Thus, $\norm{S_{\tau}}^{2}$ is bounded, up to a universal constant,
by the sum of the seven quantities
\begin{align*}
A_{1} & :=\sum_{a}\Bpar{\sum_{b}B_{b}^{\T}N_{a}NB_{b}}^{2}\,, & A_{2} & :=\sum_{a}\Bpar{\sum_{b}B_{b}^{\T}N_{b}NB_{a}}^{2}\,,\\
A_{3} & :=\sum_{a}\Bpar{\sum_{b}B_{b}^{\T}NN_{b}B_{a}}^{2}\,, & A_{4} & :=\sum_{a}\Bpar{\sum_{b}B_{b}^{\T}N^{2}B_{a,b}}^{2}\,,\\
A_{5} & :=\sum_{a}\Bpar{B_{a}^{\T}\sum_{b}N_{b}NB_{b}}^{2}\,, & A_{6} & :=\sum_{a}\Bpar{B_{a}^{\T}\sum_{b}NN_{b}B_{b}}^{2}\,,\\
A_{7} & :=\sum_{a}\Bpar{B_{a}^{\T}\sum_{b}N^{2}B_{b,b}}^{2}\,.
\end{align*}

We now estimate these terms. We use $\sum_{b}\norm{B_{b}}^{2}=\norm B_{\frob}^{2}\lesssim d$,
$\norm N,\norm{N_{a}},\norm{B^{\T}},\norm{B^{\T}N},\norm{B^{\T}N^{2}}\lesssim1$,
$\sum_{b}\norm{N_{b}B}_{\frob}^{2}=\mc E_{1}\lesssim d$, and $\mc B_{2}=\sum_{a,b}\norm{B_{a,b}}^{2}\lesssim d^{2}$.
Then,
\begin{align*}
A_{1} & \le\sum_{a}\Bpar{\sum_{b}\norm{N_{a}B_{b}}\,\norm{NB_{b}}}^{2}\lesssim\sum_{a}\norm{N_{a}}^{2}\Bpar{\sum_{b}\norm{B_{b}}^{2}}^{2}\lesssim d^{3},\\
A_{4} & \leq\sum_{a}\Bpar{\sum_{b}\norm{NB_{b}}\,\norm{NB_{a,b}}}^{2}\lesssim\sum_{a}\Bpar{\sum_{b}\norm{B_{b}}^{2}}\Bpar{\sum_{b}\norm{B_{a,b}}^{2}}\lesssim d\mc B_{2}\lesssim d^{3},\\
A_{7} & =\Bnorm{B^{\T}\sum_{b}N^{2}B_{b,b}}^{2}\lesssim\Bnorm{\sum_{b}B_{b,b}}^{2}\leq d\sum_{b}\norm{B_{b,b}}^{2}\leq d\mc B_{2}\lesssim d^{3}\,.
\end{align*}
For the remaining terms,
\begin{align*}
A_{2} & =\norm{B^{\T}N\sum_{b}N_{b}B_{b}}^{2}\lesssim\norm{\sum_{b}N_{b}B_{b}}^{2}\leq d\sum_{b}\norm{N_{b}B_{b}}^{2}\lesssim d\sum_{b}\norm{B_{b}}^{2}=d^{2}\,,\\
A_{3} & \lesssim\sum_{a}\Bpar{\sum_{b}\norm{B_{b}}^{2}}\Bpar{\sum_{b}\norm{N_{b}B_{a}}^{2}}\lesssim d\sum_{b}\norm{N_{b}B}_{\frob}^{2}\lesssim d^{2}\,,\\
A_{5} & \lesssim\Bnorm{\sum_{b}N_{b}NB_{b}}^{2}\leq d\sum_{b}\norm{N_{b}NB_{b}}^{2}\lesssim d\sum_{b}\norm{B_{b}}^{2}\lesssim d^{2}\,,\\
A_{6} & \lesssim\Bnorm{\sum_{b}NN_{b}B_{b}}^{2}\leq d\sum_{b}\norm{NN_{b}B_{b}}^{2}\lesssim d\sum_{b}\norm{B_{b}}^{2}\lesssim d^{2}\,.
\end{align*}
Combining these bounds, we prove the claim.
\end{proof}

\paragraph{Norm of $H_{3}(0)$.}

We are now ready to bound $\norm{H_{3}(0)}_{L^{2}}$.
\begin{lem}
\label{lem:quintic-l2} $\norm{H_{3}(0)}_{L^{2}}^{2}\lesssim d^{3}$,
so $H_{3}(0)=\Otilde(d^{3/2})$ with high probability.
\end{lem}

\begin{proof}
Recall that $\norm{H_{3}(0)}_{L^{2}}^{2}\lesssim\norm{T_{5}}_{\frob}^{2}+\norm{T_{3}}_{\frob}^{2}+\norm{T_{1}}^{2}$,
where $T_{5}=\sum_{i}w_{i}^{1/2}K_{i}\otimes C_{i}$ and
\[
T_{3}=\sum_{i}w_{i}^{1/2}\,(C_{i}+6K_{i}\otimes_{1}C_{i}+3K_{i}\otimes\tau_{i})\,,\qquad T_{1}=\sum_{i}w_{i}^{1/2}\,(3\tau_{i}+6K_{i}\otimes_{2}C_{i}+6K_{i}\otimes_{1}\tau_{i})\,.
\]
First of all, since $C_{i}=\sym C_{i}^{0}$, we have $\sum_{i}\norm{C_{i}}_{\frob}^{2}\leq\sum_{i}\norm{C_{i}^{0}}_{\frob}^{2}=\sum_{a,b}\norm{N_{ab}B}_{\frob}^{2}=\mc E_{2}\lesssim d^{2}$.
Following \eqref{eq:gaa-bound}, 
\[
\norm{T_{5}}_{\frob}^{2}=\sum_{ij}G_{ij}\inner{C_{i},C_{j}}\leq\sum_{i}\norm{C_{i}}_{\frob}^{2}\leq d^{2}\,.
\]

Recall that for a matrix $K$, a third-order tensor $C$, and a vector
$\tau$,
\[
(K\otimes_{1}C)_{abc}=\sum_{r}K_{ar}C_{rbc}\,,\quad(K\otimes_{2}C)_{a}=\sum_{r,s}K_{rs}C_{rsa}\,,\quad(K\otimes\tau)_{abc}=K_{ab}\tau_{c}\,,\quad(K\otimes_{1}\tau)_{a}=\sum_{r}K_{ar}\tau_{r}\,.
\]
Since $K_{i}=k_{i}k_{i}^{\T}$ and $\norm{k_{i}}=1$, we have
\[
\begin{aligned}\norm{K_{i}\otimes_{1}C_{i}}_{\frob}^{2} & =\sum_{b,c}\norm{K_{i}(C_{i})_{:bc}}^{2}\leq\sum_{b,c}\norm{(C_{i})_{:bc}}^{2}=\norm{C_{i}}_{\frob}^{2}\,,\\
\norm{K_{i}\otimes_{2}C_{i}}^{2} & =\sum_{c}\tr^{2}\bpar{K_{i}(C_{i})_{::c}}\leq\sum_{c}\norm{K_{i}}_{\frob}^{2}\norm{(C_{i})_{::c}}_{\frob}^{2}\leq\sum_{c}\norm{(C_{i})_{::c}}_{\frob}^{2}=\norm{C_{i}}_{\frob}^{2}\,.
\end{aligned}
\]

For $T_{3}$, following \eqref{eq:gaa-bound} for the last line,
\begin{align*}
\Bnorm{\sum_{i}w_{i}^{1/2}C_{i}}_{\frob}^{2} & \leq\sum_{i}w_{i}\cdot\sum_{i}\norm{C_{i}}_{\frob}^{2}\leq d\sum_{i}\norm{C_{i}}_{\frob}^{2}\leq d^{3}\,,\\
\Bnorm{\sum_{i}w_{i}^{1/2}K_{i}\otimes_{1}C_{i}}_{\frob}^{2} & \leq\sum_{i}w_{i}\cdot\sum_{i}\norm{K_{i}\otimes_{1}C_{i}}_{\frob}^{2}\leq d\sum\norm{C_{i}}_{\frob}^{2}\leq d^{3}\,,\\
\Bnorm{\sum_{i}w_{i}^{1/2}K_{i}\otimes\tau_{i}}_{\frob}^{2} & =\sum_{ij}\sqrt{w_{i}w_{j}}\inner{K_{i},K_{j}}\inner{\tau_{i},\tau_{j}}\leq\sum\norm{\tau_{i}}^{2}=\mc T_{2}\lesssim d^{3}\,.
\end{align*}

For $T_{1}$, note that $\norm{\sum_{i}w_{i}^{1/2}\tau_{i}}^{2}=\norm{S_{\tau}}^{2}\lesssim d^{3}$.
For the remaining terms,
\begin{align*}
\Bnorm{\sum_{i}w_{i}^{1/2}K_{i}\otimes_{2}C_{i}}^{2} & \leq\sum_{i}w_{i}\cdot\sum\norm{K_{i}\otimes_{2}C_{i}}^{2}\leq d\sum_{i}\norm{C_{i}}_{\frob}^{2}\lesssim d^{3}\,,\\
\Bnorm{\sum_{i}w_{i}^{1/2}K_{i}\otimes_{1}\tau_{i}}^{2} & =\Bnorm{\sum_{i}w_{i}^{1/2}K_{i}\tau_{i}}^{2}=\sum_{ij}\sqrt{w_{i}w_{j}}\inner{k_{i},k_{j}}\inner{k_{i},\tau_{i}}\inner{k_{j},\tau_{j}}\leq\sum_{i}\inner{k_{i},\tau_{i}}^{2}\leq\sum_{i}\norm{\tau_{i}}^{2}\lesssim d^{3}\,.
\end{align*}
Putting all these together, $\norm{H_{3}(0)}_{L^{2}}^{2}\lesssim d^{2}+d^{3}+d^{3}\lesssim d^{3}$.
Finally, the high-probability bound follows from Lemma~\ref{lem:conc-gaussian-poly}.
\end{proof}

\paragraph{Proof of the main claims.}
\begin{proof}
[Proof of Proposition~\ref{prop:baseline-asc}] Combining \eqref{eq:hard-recursion}
with Lemmas~\ref{lem:good-term-bounds}, \ref{lem:base-cubic}, \ref{lem:qNv-l2},
and~\ref{lem:quintic-l2}, and choosing $R_{\veps}$ suitably so
that the union of the exceptional events has probability at most $\veps$,
we get 
\[
\abs{F(\eta)-F(0)}\lesssim\eta\,\Otilde(d^{1/2})+\eta^{2}\,\Otilde(d)+\eta^{3}\,\Otilde(d^{3/2})+\eta^{4}\,\Otilde(d^{5/2})=\Otilde(r+r^{2}+r^{3}+r^{4}d^{1/2})\,.
\]
By the defining condition on $R_{\veps}$, with hidden constants small
enough, the RHS is at most $2\veps$ whenever $r\leq R_{\veps}$.
Since $\norm{y-x}_{g_{0}(y)}^{2}-\norm{y-x}_{g_{0}(x)}^{2}=\eta^{2}(F(\eta)-F(0))$,
this proves the claim.
\end{proof}
\begin{proof}
[Proof of Theorem~\ref{thm:target-mixing}] By Proposition~\ref{prop:baseline-asc},
the baseline metric satisfies ASC for $R_{\veps}=\widetilde{\Theta}(d^{-1/8})$.
By Proposition~\ref{lem:scaling-rule}, we can take $r=O(1)$ and
$g_{\star}=\widetilde{O}(d^{1/4})\,g_{0}$ while ensuring ASC. 

Recall that the unscaled LS metric is SSC, has convex log determinant,
and is $O(d\log^{3}m)$-symmetric \cite[Lemmas~4.2 and~4.3]{LLV20strong}.
Scaling the metric by $\widetilde{O}(d^{1/4})$ multiplies $\bar{\nu}$
by $\widetilde{O}(d^{1/4})$ while preserving SSC and LTSC. Thus,
$g_{\star}$ has symmetry parameter $\widetilde{O}(d^{1+1/4})$ and
satisfies SSC and LTSC. Lemma~\ref{prop:gcdw-framework} gives the
stated mixing time.
\end{proof}
\begin{acknowledgement*}
We are deeply grateful to Santosh Vempala for discussions, encouragement,
and comments on an earlier draft of this paper, and to Yuansi Chen
and Minhui Jiang for their feedback and discussions about future directions.
This work was supported in part by NSF Award CCF-2106444. We also
acknowledge the use of ChatGPT to assist with proofreading an earlier
draft. In particular, it helped with notation choices, computational
simplifications in Lemma 3.3, and the orthonormal-frame calculus in
Appendix B.
\end{acknowledgement*}
\bibliographystyle{alpha}
\bibliography{main}

\appendix

\section{Self-concordance definitions\label{app:self-concordance-definitions}}

We recall several self-concordance used in Lemma~\ref{prop:gcdw-framework}
from \cite[Definition~1.1]{KV24ipm}.
\begin{defn}[SC, SSC, and LTSC]
\label{def:sc-ssc-ltsc} Let $g:\inter\K\to\pd$ be a $C^{2}$ local
metric. We say that $g$ is \emph{self-concordant} (SC) if for every
$x\in\inter\K$ and $h\in\Rd$, 
\[
-2\norm h_{g(x)}\,g(x)\preceq\Dd g(x)[h]\preceq2\norm h_{g(x)}\,g(x)\,.
\]
We say that $g$ is \emph{strongly self-concordant} (SSC) if 
\[
\norm{g(x)^{-1/2}\Dd g(x)[h]\,g(x)^{-1/2}}_{\frob}\leq2\norm h_{g(x)}
\]
for every $x\in\inter\K$ and $h\in\Rd$. We say that $g$ is \emph{lower
trace self-concordant} (LSTC) if 
\[
\tr\bpar{g(x)^{-1}\Dd^{2}g(x)[h,h]}\geq-\norm h_{g(x)}^{2}
\]
for every $x\in\inter\K$ and $h\in\Rd$. 
\end{defn}

Note that SSC implies SC since the operator norm is bounded by the
Frobenius norm.

\section{Deferred computations}

\subsection{Tensor computations\label{app:tensor-msi-identities}}
\begin{proof}
[Proof of Lemma~\ref{lem:tensor-I2}] By Lemma~\ref{lem:msi-product},
\[
h_{i}h_{j}=I_{1}(e_{i})I_{1}(e_{j})=I_{2}(e_{i}\otimes e_{j})+e_{i}\otimes_{1}e_{j}=I_{2}(e_{i}\otimes e_{j})+\delta_{ij}\,.
\]
Thus $I_{2}(e_{i}\otimes e_{j})=h_{i}h_{j}-\delta_{ij}$. Writing
$f_{M}=\sum_{ij}M_{ij}e_{i}\otimes e_{j}$ gives 
\[
I_{2}(M)=I_{2}(f_{M})=\sum_{ij}M_{ij}(h_{i}h_{j}-\delta_{ij})=h^{\T}Mh-\tr M\,.
\]
This proves the claim.
\end{proof}
\begin{proof}
[Proof of Lemma~\ref{lem:tensor-I3}] For indices $i,j,k$, Lemma~\ref{lem:tensor-I2}
and~\ref{lem:msi-product} give
\[
h_{i}h_{j}h_{k}=\bpar{I_{2}(e_{i}\otimes e_{j})+\delta_{ij}}I_{1}(e_{k})=I_{3}(e_{i}\otimes e_{j}\otimes e_{k})+\delta_{ik}h_{j}+\delta_{jk}h_{i}+\delta_{ij}h_{k}\,.
\]
Hence, $I_{3}(e_{i}\otimes e_{j}\otimes e_{k})=h_{i}h_{j}h_{k}-\delta_{ij}h_{k}-\delta_{ik}h_{j}-\delta_{jk}h_{i}$.
Multiplying by $C_{ijk}$ and summing over $i,j,k$ gives
\[
I_{3}(C)=C[h^{\otimes3}]-\sum_{i,j,k}C_{ijk}\delta_{ij}h_{k}-\sum_{i,j,k}C_{ijk}\delta_{ik}h_{j}-\sum_{i,j,k}C_{ijk}\delta_{jk}h_{i}\,.
\]
Since $C$ is symmetric, each of the sums equals $\inner{\tr C,h}$,
where $(\tr C)_{a}=\sum_{b}C_{abb}$. Therefore $I_{3}(C)=C[h^{\otimes3}]-3\inner{\tr C,h}=C[h^{\otimes3}]-3I_{1}(\tr C)$.
\end{proof}
\begin{proof}
[Proof of Lemma~\ref{lem:tensor-I2-product}] Write $f_{M}=\sum_{ij}M_{ij}e_{i}\otimes e_{j}$
and $f_{N}=\sum_{ij}N_{ij}e_{i}\otimes e_{j}$. By Lemma~\ref{lem:msi-product},
\[
I_{2}(M)I_{2}(N)=I_{4}(f_{M}\otimes_{0}f_{N})+4I_{2}(f_{M}\otimes_{1}f_{N})+2I_{0}(f_{M}\otimes_{2}f_{N})\,.
\]
The zeroth contraction gives $I_{4}(f_{M}\otimes_{0}f_{N})=I_{4}(M\otimes N)$.
For the first contraction, 
\begin{align*}
(f_{M}\otimes_{1}f_{N})(t,s) & =\int\Bpar{\sum_{ij}M_{ij}e_{i}(t)e_{j}(x)}\,\Bpar{\sum_{kl}N_{kl}e_{k}(s)e_{l}(x)}\,\D x\\
 & =\sum_{ijk}M_{ij}N_{jk}e_{i}(t)e_{k}(s)=f_{MN}(t,s)\,.
\end{align*}
Since $I_{2}$ only sees the symmetric part of its matrix coefficient,
this contributes $I_{2}(\sym(MN))$. Finally, 
\[
f_{M}\otimes_{2}f_{N}=\int\Bpar{\sum_{ij}M_{ij}e_{i}(x)e_{j}(y)}\,\Bpar{\sum_{kl}N_{kl}e_{k}(x)e_{l}(y)}\,\D x\D y=\tr(MN)\,.
\]
Combining the three contractions proves the identity.
\end{proof}
\begin{proof}
[Proof of Lemma~\ref{lem:tensor-isometry}] This is the MSI isometry
transported through the tensor-kernel identification. Indeed, 
\[
\norm{I_{n}(M)}_{L^{2}(\Omega)}^{2}=n!\,\inner{\tilde{f}_{M},\tilde{f}_{M}}=n!\,\inner{f_{\sym M},f_{\sym M}}=n!\,\norm{\sym M}_{\frob}^{2}\,.
\]
The identity $\norm{f_{T}}_{L^{2}(\R_{+}^{n})}=\norm T_{\frob}$ follows
from the orthonormality of the functions $e_{i}$.
\end{proof}

\subsection{Row-map calculus\label{app:good-event-proof}}
\begin{lem}[Derivatives of the row map]
\label{lem:row-map} For the map $R(a):=(a\otimes a)/\norm a$ and
vectors $a\neq0$, $b$, and $c$, 
\[
\norm{\Dd R(a)[b]}\leq3\,\norm b\,,\qquad\norm{\Dd^{2}R(a)[b,c]}\lesssim\frac{\norm b\norm c}{\norm a}\,,\qquad\norm{\Dd^{3}R(a)[b,b,b]}\lesssim\frac{\norm b^{3}}{\norm a^{2}}\,.
\]
\end{lem}

\begin{proof}
By direct differentiation, 
\[
\Dd R(a)[b]=\frac{b\otimes a+a\otimes b}{\norm a}-\frac{\inner{a,b}}{\norm a^{3}}\,a\otimes a\,.
\]
Since $\norm{v\otimes w}=\norm v\norm w$, 
\[
\norm{\Dd R(a)[b]}\leq2\norm b+\norm b=3\norm b\,.
\]

For the second derivative, differentiating $\Dd R(a)[b]$ in direction
$c$, 
\begin{align*}
\Dd^{2}R(a)[b,c] & =\frac{b\otimes c+c\otimes b}{\norm a}-\frac{\inner{a,c}}{\norm a^{3}}(b\otimes a+a\otimes b)\\
 & \qquad-\frac{\inner{b,c}}{\norm a^{3}}a\otimes a-\frac{\inner{a,b}}{\norm a^{3}}(c\otimes a+a\otimes c)+3\frac{\inner{a,b}\inner{a,c}}{\norm a^{5}}a\otimes a\,.
\end{align*}
Each term is bounded by a universal multiple of $\norm b\norm c/\norm a$,
so $\norm{\Dd^{2}R(a)[b,c]}\lesssim\norm b\norm c/\norm a$. 

For the third derivative, write $R=\theta\Phi$, where $\theta(a):=\norm a^{-1}$
and $\Phi(a):=a\otimes a$. Denoting $\alpha:=\inner{a,b}$ and $\beta:=\norm b^{2}$,
and differentiating $\theta$ and $\Phi$ in direction $b$, 
\begin{align*}
\Dd\theta(a)[b] & =-\frac{\alpha}{\norm a^{3}}\,,\qquad\Dd^{2}\theta(a)[b,b]=-\frac{\beta}{\norm a^{3}}+\frac{3\alpha^{2}}{\norm a^{5}}\,,\qquad\Dd^{3}\theta(a)[b,b,b]=\frac{9\alpha\beta}{\norm a^{5}}-\frac{15\alpha^{3}}{\norm a^{7}}\,.\\
\Dd\Phi(a)[b] & =b\otimes a+a\otimes b\,,\qquad\Dd^{2}\Phi(a)[b,b]=2b\otimes b\,,\qquad\Dd^{3}\Phi(a)[b,b,b]=0\,.
\end{align*}
Thus the product rule yields
\[
\Dd^{3}R(a)[b,b,b]=\bpar{\frac{9\alpha\beta}{\norm a^{5}}-\frac{15\alpha^{3}}{\norm a^{7}}}\,a\otimes a+3\,\bpar{\frac{3\alpha^{2}}{\norm a^{5}}-\frac{\beta}{\norm a^{3}}}(b\otimes a+a\otimes b)-\frac{6\alpha}{\norm a^{3}}\,b\otimes b\,.
\]
Using $\abs{\alpha}\leq\norm a\norm b$ and $\norm{v\otimes w}=\norm v\norm w$,
we obtain $\norm{\Dd^{3}R(a)[b,b,b]}\lesssim\norm b^{3}/\norm a^{2}$.
\end{proof}

\subsection{Elementary frame calculus\label{app:moving-frame-calculus}}

The short calculation below is a standard linear-algebra fact about
choosing a ``nice'' orthonormal basis for a smoothly varying subspace.
In differential-geometric language, it is part of the standard calculus
of orthonormal frames on \emph{Stiefel/Grassmann manifolds}; see~\cite{AMS08optimization}
for references. For self-containment, we rederive exactly what we
need using only orthogonal changes of basis and differentiation of
$U^{\T}U=I$.

The next two lemmas use the same elementary idea in two settings.
Lemma~\ref{lem:path-normal-frame} treats a one-dimensional path
and ensures $U_{t}^{\T}U_{t}'=0$ for every $t$. Lemma~\ref{lem:normal-frame-calculus}
treats a multi-parameter family and records only the derivatives needed
at one base point (in our setting, when $z=x$).
\begin{lem}
\label{lem:path-normal-frame} For $t\geq0$, consider a path of smoothly
varying subspaces $E_{t}$ and an orthogonal matrix $U_{t}$ for $E_{t}$
(i.e., $\cols U_{t}=E_{t}$). There exists an orthogonal matrix path
$O_{t}$ such that $\bar{U}_{t}:=U_{t}O_{t}$ satisfies $\bar{U}_{t}^{\T}\bar{U}_{t}'=0$
and $\bar{U}_{t}'=P_{t}^{\perp}\bar{U}_{t}'$ for every $t$.
\end{lem}

\begin{proof}
Write $\bar{U}_{t}=U_{t}O_{t}$ for an orthogonal matrix $O_{t}$.
Since $\bar{U}_{t}'=U_{t}'O_{t}+U_{t}O_{t}'$,
\[
\bar{U}_{t}^{\T}\bar{U}_{t}'=O_{t}^{\T}U_{t}^{\T}(U_{t}'O_{t}+U_{t}O_{t}')=O_{t}^{\T}(U_{t}^{\T}U_{t}'O_{t}+O_{t}')\,.
\]
Hence, to ensure $\bar{U}_{t}^{\T}\bar{U}_{t}'=0$, it suffices to
enforce $O_{t}'=-U_{t}^{\T}U_{t}'O_{t}$ with $O_{0}=I$. Since $U_{t}^{\T}U_{t}=I$,
we have $(U_{t}')^{\T}U_{t}=-U_{t}^{\T}U_{t}'$. For $S_{t}:=U_{t}^{\T}U_{t}'$,
we have $S_{t}=-S_{t}^{\T}$. Then, this ODE is written as $O_{t}'=-S_{t}O_{t}$
with $O_{0}=I$. Since $S_{t}$ is smooth on $[0,\eta]$, a standard
result on the matrix ODE theory guarantees the existence of such $O_{t}$.
Moreover, $O_{t}$ remains orthogonal:
\[
\de_{t}(O_{t}^{\T}O_{t})=(O_{t}')^{\T}O_{t}+O_{t}^{\T}O_{t}'=-O_{t}^{\T}S_{t}^{\T}O_{t}-O_{t}^{\T}S_{t}O_{t}=0\,.
\]
Lastly, as $P_{t}=\bar{U}_{t}\bar{U}_{t}^{\T}$,
\[
\bar{U}_{t}'=P_{t}^{\perp}\bar{U}_{t}'+P_{t}\bar{U}_{t}'=P_{t}^{\perp}\bar{U}_{t}'+\bar{U}_{t}\bar{U}_{t}^{\T}\bar{U}_{t}'=P_{t}^{\perp}\bar{U}_{t}'\,,
\]
which proves the claim.
\end{proof}
\begin{lem}
\label{lem:normal-frame-calculus} Let $P(v)$ be a smooth orthogonal
projector, defined near $v=0$, and let $\widetilde{U}(v)$ be smooth
orthogonal matrix with $\cols\widetilde{U}(v)=\operatorname{im}P(v)$.
There exists an orthogonal matrix $O(v)$ such that $U:=\widetilde{U}O$
satisfies at $v=0$ that \textup{for every coordinate direction $a,b$,}
\[
U^{\T}U_{a}=0\,,\qquad U_{a}=P^{\perp}U_{a}=P_{a}U\,,\qquad P^{\perp}U_{ab}=P^{\perp}(P_{a}U)_{b}\,.
\]
This also satisfies $U^{\T}U_{ab}=-\frac{1}{2}\,(U_{b}^{\T}U_{a}+U_{a}^{\T}U_{b})$
and
\[
\sum_{a,b}\norm{U_{ab}}_{\frob}^{2}\lesssim\sum_{a,b}\norm{P^{\perp}U_{ab}}_{\frob}^{2}+\Bpar{\sum_{a}\norm{U_{a}}_{\frob}^{2}}^{2}\,.
\]
\end{lem}

\begin{proof}
All quantities in this proof are evaluated at $v=0$ unless stated
otherwise. Proof for the first-order condition is similar to the previous
lemma in overall. Since $\widetilde{U}^{\T}\widetilde{U}=I$, differentiating
gives $\widetilde{U}^{\T}\widetilde{U}_{a}+\widetilde{U}_{a}^{\T}\widetilde{U}=0$,
and $\Omega_{a}:=\widetilde{U}^{\T}\widetilde{U}_{a}$ is skew-symmetric.
Choose $O(v):=\exp(-\sum_{c}v_{c}\Omega_{c}(0))$. Then, $O(0)=I$
and $O_{a}(0)=-\Omega_{a}(0)$. For $U:=\widetilde{U}O$, we can check
that $U^{\T}U_{a}=0$ and $U_{a}=P^{\perp}U_{a}$ in a similar fashion
to Lemma~\ref{lem:path-normal-frame}. Lastly, since $PU=U$, we
have $P_{a}U+PU_{a}=U_{a}$. This implies $P_{a}U=P^{\perp}U_{a}$.

For second derivatives, note that $U_{a}=P^{\perp}U_{a}+PU_{a}=P_{a}U+PU_{a}=P_{a}U+UU^{\T}U_{a}$.
Defining $\Omega_{a}:=U^{\T}U_{a}$ and differentiating $U_{a}=P_{a}U+U\Omega_{a}$
in direction $b$ at $v=0$,
\[
U_{ab}=(P_{a}U)_{b}+U(\Omega_{a})_{b}+U_{b}\Omega_{a}=(P_{a}U)_{b}+U(\Omega_{a})_{b}\,.
\]
This implies that $P^{\perp}U_{ab}=P^{\perp}(P_{a}U)_{b}$. 

For $\Gamma_{ab}:=U^{\T}U_{ab}$, write 
\begin{equation}
U_{ab}=P^{\perp}U_{ab}+PU_{ab}=P^{\perp}U_{ab}+U\Gamma_{ab}\,.\label{eq:second-deriv}
\end{equation}
Just as we enforced $U^{\T}U_{a}=0$, we will enforce some condition
on $U$ so that $\Gamma_{ab}$ is symmetric, while preserving $U_{a}$
and ensuring $U^{\T}U_{a}=0$. We will replace $U(v)$ by $\widehat{U}(v):=U(v)O(v)$
with $O(v)=\exp(\frac{1}{2}\sum_{a,b}v_{a}v_{b}K_{ab})$ where $K_{ab}=K_{ba}$
and each $K_{ab}$ is skew-symmetric. Clearly, $O(0)=I$, $O_{a}(0)=0$,
and $O_{ab}(0)=K_{ab}$. Differentiating $\widehat{U}=UO$ and estimating
it at $v=0$ yields 
\begin{align*}
\widehat{U}_{a} & =U_{a}O+UO_{a}=U_{a}\,,\\
\widehat{U}_{ab} & =U_{ab}O+U_{a}O_{b}+U_{b}O_{a}+UO_{ab}=U_{ab}+UK_{ab}\,.
\end{align*}
Hence, $\widehat{U}^{\T}\widehat{U}_{a}=0$ and $\widehat{U}_{a}=P^{\perp}\widehat{U}_{a}$
remain valid. Next, from above, $P^{\perp}\widehat{U}_{ab}=P^{\perp}U_{ab}$
and 
\[
P^{\perp}(P_{a}\widehat{U})_{b}=P^{\perp}(P_{a}UO)_{b}=P^{\perp}(P_{a}U)_{b}\,,
\]
so $P^{\perp}\widehat{U}_{ab}=P^{\perp}(P_{a}\widehat{U})_{b}$.

Note that $\widehat{\Gamma}_{ab}=\widehat{U}^{\T}\widehat{U}_{ab}=U^{\T}\,(U_{ab}+UK_{ab})=\Gamma_{ab}+K_{ab}$.
Since $\Gamma_{ab}=\Gamma_{ba}$, choosing $K_{ab}=-\operatorname{skew}(\Gamma_{ab})$
is compatible with $K_{ab}=K_{ba}$ and makes $\widehat{\Gamma}_{ab}$
symmetric. From now on, we relabel $\widehat{U}$ as $U$ and $\widehat{\Gamma}_{ab}$
as $\Gamma_{ab}$. Differentiating $U^{\T}U=I$ in directions $a$
and $b$ and estimating it at $v=0$,
\[
U_{b}^{\T}U_{a}+U^{\T}U_{ab}+U_{ab}^{\T}U+U_{a}^{\T}U_{b}=0\,.
\]
Since $\Gamma_{ab}=U^{\T}U_{ab}$, we have $\Gamma_{ab}+\Gamma_{ab}^{\T}=-U_{b}^{\T}U_{a}-U_{a}^{\T}U_{b}$.
Then, since $\Gamma_{ab}$ is symmetric,
\[
\Gamma_{ab}=-\frac{1}{2}\,(U_{b}^{\T}U_{a}+U_{a}^{\T}U_{b})\,.
\]
Therefore,
\[
\sum_{a,b}\norm{\Gamma_{ab}}_{\frob}^{2}\lesssim\sum_{a,b}\norm{U_{a}}_{\frob}^{2}\norm{U_{b}}_{\frob}^{2}=\Bpar{\sum_{a}\norm{U_{a}}_{\frob}^{2}}^{2}\,.
\]
The final item follows from $U_{ab}=P^{\perp}U_{ab}+U\Gamma_{ab}$
\eqref{eq:second-deriv}.
\end{proof}

\end{document}